\documentclass[11pt]{article}
\usepackage{amsthm}
\usepackage{amsmath, amssymb, amsfonts, harvard,color, graphicx, mathrsfs, bm,  graphics,float,tikz}
\usepackage {multirow}
\usepackage[dvips]{epsfig}

\newcommand{\bu}{\mbox{$\bm{u}$}}
\newcommand{\bv}{\mbox{$\bm{v}$}}
\newcommand{\bw}{\mbox{$\bm{w}$}}
\newcommand{\bx}{\mbox{$\bm{x}$}}
\newcommand{\by}{\mbox{$\bm{y}$}}
\newcommand{\bz}{\mbox{$\bm{z}$}}
\newcommand{\bX}{\mbox{$X$}}
\newcommand{\bY}{\mbox{$Y$}}

\newcommand{\bl}{\mbox{$\Lambda$}}
\newcommand{\bt}{\mbox{$\Theta$}}

\newcommand{\bs}{\mbox{$\Sigma$}}

\newcommand{\tls}{\mbox{$\tilde{\Lambda}$}}
\newcommand{\tts}{\mbox{$\tilde{\Theta}$}}

\newcommand{\bphi}{\mbox{$\bm{\phi}$}}
\newcommand{\bxi}{\mbox{$\bm{\xi}$}}
\newcommand{\bdd}{\mbox{$\bm{d}$}}
\newcommand{\bcc}{\mbox{$\bm{c}$}}

\newcommand{\sx}{\mbox{$S_{xx}$}}
\newcommand{\sy}{\mbox{$S_{yy}$}}
\newcommand{\sxy}{\mbox{$S_{xy}$}}

\newcommand{\sxk}{\mbox{$S_{xx,k}$}}
\newcommand{\syk}{\mbox{$S_{yy,k}$}}
\newcommand{\sxyk}{\mbox{$S_{xy,k}$}}

\newcommand{\sxkk}{\mbox{$S_{xx}^{(-k)}$}}
\newcommand{\sykk}{\mbox{$S_{yy}^{(-k)}$}}
\newcommand{\sxykk}{\mbox{$S_{xy}^{(-k)}$}}

\newcommand{\hll}{\mbox{$\hat{\Lambda}$}}
\newcommand{\htt}{\mbox{$\hat{\Theta}$}}
\newcommand{\hlk}{\mbox{$\hat{\Lambda}^{(-k)}$}}

\newcommand{\htk}{\mbox{$\hat{\Theta}^{(-k)}$}}

\newcommand{\tr}{\mbox{$\text{tr}$}}
\newcommand{\trans}{\mbox{$^T$}}
\newcommand{\inv}{\mbox{$^{-1}$}}
\newcommand\norm[1]{\left\lVert#1\right\rVert}
\newcommand{\vecc}{\mbox{$\mathrm{vec}$}}

\newcommand{\vect}{\mbox{$\text{vec}$}}

\newcommand{\cc}{\mbox{$C_{\sigma}$}}
\newcommand{\ccs}{\mbox{$C^2_{\sigma}$}}
\newcommand{\cx}{\mbox{$C_{X}$}}

\newcommand{\intx}{\mbox{$\int_{\mathcal{X}}$}}
\newcommand{\etah}{\mbox{$\hat{\eta}$}}
\newcommand{\zetah}{\mbox{$\hat{\zeta}$}}
\newcommand{\etat}{\mbox{$\tilde{\eta}$}}
\newcommand{\zetat}{\mbox{$\tilde{\zeta}$}}
\newcommand{\etau}{\mbox{$\eta_u$}}
\newcommand{\zetau}{\mbox{$\zeta_u$}}

\newcommand{\rhox}{\mbox{$\rho(\bx)$}}
\newcommand{\dx}{\mbox{$d\bx$}}

\newcommand{\vertiii}[1]{{\left\vert\kern-0.25ex\left\vert\kern-0.25ex\left\vert #1 
		\right\vert\kern-0.25ex\right\vert\kern-0.25ex\right\vert}}

\makeatletter
\newcommand*{\rom}[1]{\expandafter\@slowromancap\romannumeral #1@}
\makeatother

\DeclareMathOperator*{\argminA}{arg\,min}
\DeclareMathOperator*{\argmaxA}{arg\,max}

\newcommand{\redline}{\raisebox{2pt}{\tikz{\draw[-,red,solid,line width = 1pt](0,0) -- (5mm,0);}}}
\newcommand{\blueline}{\raisebox{2pt}{\tikz{\draw[-,blue,dash dot,line width = 1pt](0,0) -- (5mm,0);}}}
\usepackage[ruled,vlined]{algorithm2e}

\SetKwInput{KwInput}{Input}
\SetKwInput{KwOutput}{Output}
\SetKwInput{KwIni}{Initialize}

\usepackage{scalerel,stackengine}
\stackMath
\newcommand\reallywidehat[1]{%
	\savestack{\tmpbox}{\stretchto{%
			\scaleto{%
				\scalerel*[\widthof{\ensuremath{#1}}]{\kern-.6pt\bigwedge\kern-.6pt}%
				{\rule[-\textheight/2]{1ex}{\textheight}}%WIDTH-LIMITED BIG WEDGE
			}{\textheight}% 
		}{0.5ex}}%
	\stackon[1pt]{#1}{\tmpbox}%
}
\newtheorem{theorem}{Theorem}
\newtheorem{assumption}{Assumption}
\newtheorem{lemma}{Lemma}
\newtheorem{corollary}{Corollary}

\begin{document}
	
	\title{Combining Smoothing Spline with Conditional Gaussian Graphical Model for 
		Density and Graph Estimation}
	%Semiparametric Density and Graph Estimation Using Smoothing Spline
	%and Conditional Gaussian Graphical Models \\
	%Multivariate Density Estimation and Graphical Models
	\author{Runfei Luo, Anna Liu and Yuedong Wang
		\thanks{Runfei Luo (email: rluo@pstat.ucsb.edu) received her Ph.D. in statistics from University of California, Santa Barbara. She is now an applied scientist in Amazon Web Services.
			{Anna Liu (email: anna@math.umass.edu) is Associate Professor, Department of
			Mathematics and Statistics, University of Massachusetts, Amherst, Massachusetts 01002.},
			Yuedong Wang (email: yuedong@pstat.ucsb.edu) is Professor, Department 
			of Statistics and Applied Probability, University of California, 
			Santa Barbara, California 93106. 
			Anna Liu's research was supported by a grant from the National Science
			Foundation (DMS-1507078). 
			Runfei Luo and Yuedong Wang's research was supported by a grant from the National Science
			Foundation (DMS-1507620).
			We acknowledge support from the Center for Scientific Computing from 
			the CNSI, MRL: an NSF MRSEC (DMR-1720256) for their support.
			Address for correspondence: Yuedong Wang, 
			Department of Statistics and Applied Probability, University of 
			California, Santa Barbara, California 93106.}}
	
	\maketitle
	
	\begin{abstract}
		
		Multivariate density estimation and graphical models play important roles 
		in statistical learning. The estimated density can be used to construct a 
		graphical model that reveals conditional relationships whereas a graphical 
		structure can be used to build models for density estimation. Our goal is 
		to construct a consolidated framework that can perform both density and 
		graph estimation. 
		Denote $\bm{Z}$ as the random vector of interest with density function 
		$f(\bz)$. Splitting $\bm{Z}$ into two parts, $\bm{Z}=(\bm{X}^T,\bm{Y}^T)^T$ 
		and writing $f(\bz)=f(\bx)f(\by|\bx)$ where $f(\bx)$ is the density function 
		of $\bm{X}$ and $f(\by|\bx)$ is the conditional density of $\bm{Y}|\bm{X}=\bx$. 
		We propose a semiparametric framework that models $f(\bx)$ nonparametrically 
		using a smoothing spline ANOVA (SS ANOVA) model and $f(\by|\bx)$ parametrically 
		using a conditional Gaussian graphical model (cGGM). Combining flexibility of 
		the SS ANOVA model with succinctness of the cGGM, this framework allows us to 
		deal with high-dimensional data without assuming a joint Gaussian distribution.  
		We propose a backfitting estimation procedure for the cGGM with 
		a computationally efficient approach for selection of tuning parameters. We also 
		develop a geometric inference approach for edge selection. We establish 
		asymptotic convergence properties for both the parameter and density estimation.
		The performance of the proposed method is evaluated through 
		extensive simulation studies and two real data applications.
		
		
	\end{abstract}
	
	\section{Introduction}
	
	Density estimation has long been a subject of paramount interest in statistics.
	Many parametric, nonparametric, and semiparametric methods have been developed
	in the literature. 
	Assuming a known distribution family with succinct representation and 
	interpretable parameters, the parametric approach is in general statistically and 
	computationally efficient \cite{Kendall1987}. 
	However, the parametric assumption may be too 
	restrictive for some applications. The nonparametric approach, on the other hand, 
	does not assume a specific form for the density function and allows its shape 
	to be decided by data. Methods such as kernel estimation 
	\cite{parzen1962estimation,silverman2018density}, local likelihood estimators 
	\cite{loader1996local}, and smoothing splines 
	\cite{gu2013smoothing} work well for low dimensional multivariate density 
	functions. When the dimension is moderate to large, existing nonparametric methods 
	break down quickly due to the curse of dimensionality and/or computationally 
	limitations. \citeasnoun{duong2007ks} pointed out that the kernel density estimation is 
	not applicable to random variables of dimension higher than six.
	To reduce the computational burden, \citeasnoun{jeon2006effective} and 
	\citeasnoun{gu2013nonparametric} developed pseudo likelihood method for smoothing 
	spline density estimation. However, our experience indicates that the computation 
	become almost infeasible when the dimension is higher than twelve.
	Consequently, contrary to the univariate case, flexible methods for multivariate 
	density estimation are rather limited when the dimension is large.
	Recent work using piecewise constant and Bayesian partitions represents a
	major breakthrough in this area 
	\cite{lu2013multivariate,liu2014multivariate,li2016density}. 
	Nevertheless, these methods can handle moderate 
	dimensions only, lead to non-smooth density estimates, and cannot be used to 
	investigate the conditional relationship. 
	
	Some semiparametric methods have been proposed to take advantage of the parsimony
	of parametric models and the flexibility of nonparametric modeling. 
	Semiparametric copula models consist of nonparametric marginal distributions and 
	parametric copula functions \cite{genest1995}. 
	Projection pursuit density estimation overcomes the curse of dimensionality by 
	representing the joint density as a product of some smooth univariate functions of 
	carefully selected linear combinations of variables \cite{friedman1984projection}.
	The regularized derivative expectation operator (rodeo) method assumes the joint density 
	equals a product of a parametric component and a nonparametric function of 
	an unknown subset of variables \cite{liu2007sparse}.
	%({\color{red} rodeo performs bandwidth selection and variable selection 
	%simultaneously, thus it can be used to select variables for nonparametric 
	%modeling in our method. Can we find redeo code or code ourself?})
	Other semiparametric/nonparametric methods for density estimation 
	include mixture models \cite{richardson1997bayesian}, forest density 
	\cite{liu2011forest}, density tree \cite{Ram11}, and
	geometric density estimation \cite{Dunson16}.
	All existing semiparametric/nonparametric methods have strengths and limitations.
	We will develop a new semiparametric procedure for multivariate density 
	estimation that explores the sparse graph structure in the parametric part of the model. 
	
	Graphical models are used to characterize conditional relationship between variables 
	with a wide range of applications in natural sciences,
	social sciences, and economics 
	\cite{lauritzen1996graphical,fan2016overview,friedman2008sparse}.
	Gaussian graphical model (GGM) is one of the most popular models where
	conditional independence is reflected in the zero entries of the precision matrix
	\cite{friedman2008sparse}. The resulting structure from a GGM can be erroneous when 
	the true distribution is far from Gaussian.
	The dependence structure of non-Gaussian data has not received great attention 
	until recent years. Robustified Gaussian and elliptical graphical models against 
	possible outliers were studied by \citeasnoun{miyamura2006robust}, 
	\citeasnoun{finegold2011robust}, \citeasnoun{vogel2011elliptical}, and 
	\citeasnoun{sun2012robust}. Graphical models based on generalized linear models
	were proposed by \citeasnoun{lee2007efficient}, \citeasnoun{hofling2009estimation}, 
	\citeasnoun{ravikumar2010high}, \citeasnoun{allen2012log}, and 
	\citeasnoun{yang2012graphical}.
	Nonparametric and semiparametric approaches have also been considered. 
	\citeasnoun{jeon2006effective} and \citeasnoun{gu2013nonparametric} applied SS ANOVA 
	dendity models to estimate graphs (see Section \ref{sec:graphestimation} 
	for details). The computation 
	of this nonparametric approach becomes prohibitive for large dimensions. 
	\citeasnoun{liu2009nonparanormal}, \citeasnoun{liu2012high}, and 
	\citeasnoun{xue2012regularized} developed an elegant nonparanormal model which 
	assumes that there exists a monotone transformation to each variable such that the 
	joint distribution after transformation is multivariate Gaussian. Then any 
	established estimation methods for the GGM can be applied to the transformed variables.
	Other semiparametric/nonparametric methods include graphical random forests 
	\cite{fellinghauer2013stable}, regularized score matching \cite{lin2018methods}, 
	and kernel partial correlation \cite{oh2017graphical}.

	The goal of this article is to build a semiparametric model that combines 
	the GGM with the SS ANOVA density model. We are interested in both density and 
	graph estimation. 
	The remainder of the article is organized as follows. In Section
	2 we introduce the semiparametric density model and methods for estimation 
	and computation. We propose methods for graph estimation in Section 3. 
	Sections 4 presents theoretical properties of our methods in term of 
	both density and graph estimation.  
	In Section 5 we evaluate our method using simulation studies. 
	In Section 6 we present applications to two real datasets.
	Some technical details are gathered in the Appendix.

	\section{Density Estimation with SS ANOVA and cGGM}
	\label{sec:densityestimation}

	\subsection{Semiparametric Density Models with SS ANOVA and cGGM}
	\label{sec:2.1}
	
	Consider the density estimation problem in which we are given a random 
	sample of a random vector $\bm{Z}$, and we wish to estimate the density 
	function $f(\bm{z})$ of $\bm{Z}$. Let $\bm{Z} = (\bm{X}\trans,\bm{Y}\trans)\trans$ 
	where $\bm{X}=(X_1, \cdots, X_d)^T$ is a $d$-dimensional random vector for which the density 
	function will be modeled nonparametrically and $\bm{Y}=(Y_1,\cdots, Y_p)^T\in\mathbb{R}^p$ 
	collects elements for which the conditional density will be modeled 
	parametrically. The joint density function $f(\bz)$ can be decomposed 
	into two components:
	\begin{equation}
	f(\bz) = f(\bx,\by) = f(\bx)f(\by|\bx).
	\label{eq:2.1.1}
	\end{equation}
	We will model $f(\bx)$ and $f(\by|\bx)$ using SS ANOVA models and 
	cGGMs respectively. We now provide details of these models. 
	
	Assume $\bm{X} \in \mathcal{X}=\mathcal{X}_{1}\times\dots\times\mathcal{X}_{d}$ where
	 $X_u \in \mathcal{X}_{u}$ which is an arbitrary set. 
	To deal with the positivity and unity constraints of a density function, 
	we consider the logistic transform $f = e^\eta / \intx e^\eta d\bx$ 
	where $\eta(\bx)$ is referred to as the logistic density function 
	\cite{gu2013smoothing}.
	We construct a model space for $\eta$ using the tensor product of reproducing 
	kernel Hilbert spaces (RKHS). The SS ANOVA decomposition of functions 
	in the tensor product RKHS can be represented as 
	\begin{equation}
	\eta(\bx) = c + \sum_{k=1}^{d}\eta_k(x_k) + \sum_{k>l}\eta_{kl}(x_k, x_l) + \dots
	+\eta_{1\ldots d}(x_1,\cdots,x_d),
	\label{eq:ssanovadecomp}
	\end{equation}
	where $\eta_k$'s are main effects, $\eta_{kl}$'s are two-way interactions, 
	and the rest are higher order interactions involving more than two variables. 
	Higher order interactions are often removed in (\ref{eq:ssanovadecomp}) for more 
	tractable estimation and inference. 
	An SS ANOVA model for the logistic density function assumes that $\eta$ 
	belongs to an RKHS which contains a subset of components in the SS ANOVA 
	decomposition (\ref{eq:ssanovadecomp}). For a given SS ANOVA model, terms included 
	in the model can be regrouped and the model space can be expressed as
	\begin{equation}
	\mathcal{H} = \mathcal{H}^0\oplus\mathcal{H}^1\oplus\cdots\oplus\mathcal{H}^w,
	\label{eq:ssanova}
	\end{equation}
	where $\mathcal{H}^0$ is a finite dimensional space collecting all functions 
	that are not going to be penalized, and $\mathcal{H}^1,\dots,\mathcal{H}^w$ 
	are orthogonal RKHS's with reproducing kernels (RK) $R^v$ for $v=1\dots,w$.
	Details about the SS ANOVA model can be found in \citeasnoun{gu2013smoothing} 
	and \citeasnoun{wang2011smoothing}.
	
	We assume a cGGM for $f(\by|\bx)$. Specifically, we assume that 
	$\bm{Y} | \bm{X}=\bx \sim \text{N}(-\Lambda\inv\Theta\trans\bx,\Lambda\inv)$
	where $\Lambda$ is a $p \times p$ precision matrix and 
	$\Theta$ is a $d\times p$ matrix that parameterizes the conditional 
	relationship between $\bm{X}$ and $\bm{Y}$ 
	\cite{sohn2012joint,wytock2013sparse,yuan2014partial}. We note that 
	the negative log likelihood function is convex under this parameterization.
	An alternative assumption 
	$\bm{Y} | \bm{X}=\bx \sim \text{N}(\Psi\bx, \Lambda\inv)$ \cite{yin2011sparse}
	may be used to model the conditional density $f(\by|\bx)$ where 
	the negative log likelihood function is biconvex in $\Psi$ and $\Lambda$ 
	rather than jointly convex. 
	
	We will refer to the proposed semiparametric model as combined smoothing spline 
	and conditional Gaussian graphical (cSScGG) model. The cSScGG model
	is closely related to the semiparametric 
	kernel density estimation (SKDE) proposed by \citeasnoun{hoti2004semiparametric}. 
	The same decomposition in (\ref{eq:2.1.1}) was considered. Given an iid sample 
	$\bm{Z}_i = (\bm{X}_i\trans,\bm{Y}_i\trans)\trans$, $i=1,\dots,n$, 
	\citeasnoun{hoti2004semiparametric} estimated $f(\bx)$ using the kernel density, 
	$\hat{f}(\bx) = n^{-1}\sum_{i=1}^{n}K_{h_1}(\bx - \bm{X}_i)$, and
	$f(\by|\bx)$ using the conditional Gaussian density with mean $\mu(\bx)$ and 
	covariance $\Sigma(\bx)$. Specifically, they estimated $\mu(\bx)$ and 
	covariance $\Sigma(\bx)$ by 
	$\hat{\mu}(\bx) = \sum_{i=1}^nW_{h_2}(\bx-\bm{X}_i)\bm{Y}_i$ and 
	$\hat{\Sigma}(\bx) = \sum_{i=1}^{n}W_{h_3}(\bx-\bm{X}_i)
	(\bm{Y}_i - \hat{\mu}(\bx))(\bm{Y}_i - \hat{\mu}(\bx))\trans$ respectively,
	where 
	$K_{h}(\bx) = h^{-d}K(\bx/h)$, 
	$K$ is the symmetric Gaussian kernel function, 
	$W_h(\bm{x}-\bm{X}_i) = K_h(\bx-\bm{X}_i)/\sum_{j=1}^{n}K_h(\bx-\bm{X}_i)$,
	and $h_1$, $h_2$, and $h_3$ are bandwidths. Selection of bandwidths can be 
	difficult and the estimation of conditional mean and covariance can be poor
	when the dimension of $\bm{Y}$ is large. 
	The authors focused on the classification problem. They set 
	$W_{h_2}(\bx-\bm{X}_i)=W_{h_3}(\bx-\bm{X}_i)=1/n$ in their simulations to 
	make the computation feasible. Under these weights the estimated conditional 
	density $f(\bm{y}|\bm{x})$ does not depend on $\bm{x}$ at all. In contrast, 
	we model $f(\bm{y}|\bm{x})$ using a cGGM which will allow us to explore 
	sparsity in the conditional dependence structure. In addition, the domain $\mathcal{X}$ in our model is an arbitrary set while the domain in the SKDE method is a subset of $\mathbb{R}^d$. While we focus on continuous $\bm{X}$ in this paper, the discrete case is a natural extension of the current work. 	
	
	\subsection{Penalized Likelihood Estimation}
	\label{sec:2.2}
	
	A cSScGG model consists of three parameters: $\eta \in \mathcal{H}$ and
	matrices $\Lambda$ and $\Theta$ where $\mathcal{H}$ is an RKHS given in 
	\eqref{eq:ssanova}
	and $\Lambda$ is positive definite. 
	Given an iid sample 
	$\bm{Z}_i = (\bm{X}_i\trans,\bm{Y}_i\trans)\trans$, $i=1,\dots,n$,
	let $X=(\bm{X}_1, \dots, \bm{X}_n)\trans$,  
	$Y=(\bm{Y}_1, \dots, \bm{Y}_n)\trans$, 
	$\sx = n^{-1}\bX\trans\bX$, 
	$\sy = n^{-1}\bY\trans\bY$,
	and $\sxy = n^{-1}\bX\trans\bY$. Denote 
	\begin{eqnarray}
	l_1(\eta)&=& \frac{1}{n}\sum_{i=1}^{n}e^{-\eta(\bm{X}_i)} + 
	\int_\mathcal{X}\eta(\bx)\rho(\bx)d\bx, \label{eq:l1} \\
	l_2(\bt,\bl)&=& -\log|\bl| + \tr(\sy\bl + 2S_{xy}^{T}\bt + 
	\bl\inv\bt\trans S_{xx}\trans\bt) \label{eq:l2}
	\end{eqnarray}
	as the negative log pseudo likelihood and negative log likelihood functions based on 
	$\bm{X}$ and $\bm{Y}$ samples respectively, where some constants are ignored and
	$\rho$ is a known density for the pseudo likelihood \cite{gu2013smoothing}.
	The function $l_1(\eta)$ is
	continuous, convex and Fr\'{e}chet differentiable \cite{jeon2006effective}, and 
	the function $l_2(\bt,\bl)$ is jointly convex in $\bl$ and $\bt$. 
	
	We estimate $\eta$, $\Lambda$ and $\Theta$ as minimizers of the 
	penalized likelihood:
	\begin{equation}
	\{\hat{\eta}, \hat{\bl}, \hat{\bt} \} = 
	\argminA_{\eta\in\mathcal{H}, \Lambda \succ 0, \Theta} 
	\left\{ \Big[ l_1(\eta) + \frac{\lambda_1}{2}J(\eta) \Big]
	+ \Big[ l_2(\bl,\bt) + \lambda_2 \norm{\bl}_{1,\text{off}} + 
	\lambda_3 \norm{\bt}_1 \Big ] \right\}, 
	\label{eq:penlik}
	\end{equation} 
	where $J$ is a semi-norm in $\mathcal{H}$ that penalizes departure from the null space 
	$\mathcal{H}^0$,  $\norm{\cdot}_1$ denotes the elementwise $\ell_1$-norm, 
	$\norm{\cdot}_{1,\text{off}}$ denotes the elementwise $\ell_1$-norm on 
	off-diagonal entries, and $\Lambda \succ 0$ indicates positive definiteness of 
	$\Lambda$. Together, $\norm{\bl}_{1,\text{off}}$ and $\norm{\bt}_1$ encourage 
	sparsity for the cGGM. We allow different tuning parameters for different 
	penalties. 
	
	Note that the first part of the penalized likelihood depends on $\eta$ only 
	and the second part depends on $\Theta$ and $\Lambda$ only. Therefore, we can 
	compute the penalized likelihood estimates by solving two optimization problems
	separately:
	\begin{equation}
	\hat{\eta} = \argminA_{\eta\in\mathcal{H}} \left\{ \frac{1}{n}\sum_{i=1}^{n}e^{-\eta(\bm{X}_i)} + \int_\mathcal{X}\eta(\bx)\rho(\bx)d\bx + \frac{\lambda_1}{2}J(\eta)\right\},
	\label{eq:pl1}
	\end{equation}
	and 
	\begin{equation}
	\begin{aligned}
	\{\hat{\bt}, \hat{\bl} \} = &\argminA_{\Lambda \succ 0, \Theta} \left\{ -\log|\bl| + \tr(\sy\bl + 2S_{xy}^{T}\bt + \bl\inv\bt\trans S_{xx}\trans\bt) + \lambda_2 \norm{\bl}_{1,\text{off}} + \lambda_3 \norm{\bt}_1 \right\}. 
	\label{eq:pl2}
	\end{aligned}
	\end{equation}
	
	As in \citeasnoun{gu2013smoothing}, we approximate the solution of \eqref{eq:pl1} 
	by a linear combination of basis functions in $\mathcal{H}^0$ and a random subset
	of representers. Then the estimate $\hat{\eta}$ can be calcuated using the 
	Newton-Raphson algorithm. The smoothing parameter $\lambda_1$ is selected as the
	minimizer of an approximated cross-validation estimate of the Kullback-Leibler (KL) 
	divergence. Details can be found in \citeasnoun{gu2013smoothing}, 
	\citeasnoun{gu2013nonparametric}, and \citeasnoun{Luothesis}. 
	In the next section we propose a new computational method
	for solving \eqref{eq:pl2}.

	\subsection{Backfitting Algorithm for cGGM}
	\label{sec:2.3.2}
	
	%Sohn and Kim \cite{sohn2012joint} and Wytock and Kolter \cite{wytock2013sparse} 
	%respectively proposed orthant-wise quasi-Newton, and second-order active-set 
	%algorithms to estimate $\Lambda$ and $\Theta$ simultaneously. These algorithms 
	%can obtain the global optimal estimators. However, updating $\Lambda$ and 
	%$\Theta$ simultaneously needs a full Hessian matrix, which is computationally 
	%expensive and does not scale to a high-dimensional problem. To bypass this 
	%issue, Yuan and Zhang \cite{yuan2014partial} utilized a smoothing approximation 
	%procedure to optimize the approximated objective function. 
	
	Instead of updating $\Lambda$ and $\Theta$ simultaneously as in
	\citeasnoun{sohn2012joint}, \citeasnoun{wytock2013sparse} and \citeasnoun{yuan2014partial}, we will consider a 
	backfitting procedure to update them iteratively until convergence. We use the
	subscript $(t)$ to denote quantities calculated at iteration $t$ and $A_{ij}$ 
	to denote the $(i,j)$-th element of a matrix $A$. 
	
	At iteration $t+1$, with $\Lambda$ being fixed at $\bl_{(t)}$, \eqref{eq:pl2} reduces to the 
	minimization of a quadratic function plus an $\ell_1$ penalty. Therefore, 
	without needing to calculate the Hessian matrix, $\Theta$ can be updated 
	efficiently using the coordinate descent algorithm. The gradient
	$\nabla_\Theta l_2(\Lambda,\Theta) = 2S_{xy} + 2S_{xx}\Theta\Lambda^{-1}$.
	Denote $\Sigma=\Lambda^{-1}$ as the covariance matrix.
	Then the $(i,j)$th element $\Theta_{ij}$ is updated by
	%\Theta_{ij,(t)} - c_{\Theta} + 
	\begin{equation}
	\bt_{ij,(t+1)} \leftarrow  S_{\lambda_3 / a_{\Theta}}\Big(c_{\Theta} - \frac{b_{\Theta}}{a_{\Theta}}\Big),
	\label{eq:Thetaupdate}
	\end{equation}
	where $a_{\Theta} = 2\bs_{jj,(t)}(\sx)_{ii}$, 
	$b_{\Theta} = 2(\sxy)_{ij} + 2(\sx\Theta_{(t)}\bs_{(t)})_{ij}$, 
	$c_{\Theta} = \Theta_{ij,(t)}$, and 
	$S_\omega(x) = \text{sign}(x)\max(|x| - \omega, 0)$ is the soft-thresholding operator
	with threshold $\omega$.
	
	To update $\Lambda$ at iteration $t+1$, we consider the approximate
	conditional distribution 
	$\text{N}(-\Lambda^{-1}_{(t)}\Theta\trans_{(t)}\bx,\Lambda^{-1})$
	where both $\Theta$ and $\Lambda$ in the conditional mean are fixed at their 
	estimates from the $t$-th iteration. The resulting negative log likelihood 
	\begin{equation}
	h_{(t)}(\Lambda) = -\log|\Lambda| + 
	\tr\big(\sy\Lambda + 2S_{xy}^{T}\Theta_{(t)}\Lambda^{-1}_{(t)}\Lambda + 
	\Lambda^{-1}_{(t)}\Lambda\Lambda^{-1}_{(t)}\Theta\trans_{(t)}\sx\Theta_{(t)}\big)
	\end{equation}
	where terms independent of $\Lambda$ are dropped. We update $\Lambda$ by
	\begin{equation}
	\Lambda_{(t+1)} = \argminA_{\Lambda \succ 0}\Big\{h_{(t)}(\Lambda)
	+  \lambda_2 \norm{\Lambda}_{1,\text{off}} \Big\} .
	\label{eq:Lambdaupdate}
	\end{equation}
	
	As in \citeasnoun{hsieh2011sparse}, we will find the Newton direction by approximating
	$h_{(t)}$ using a quadratic function. Based on the second-order Taylor expansion of
	$h_{(t)}(\Lambda)$ at $\Lambda_{(t)}$ where $\Lambda = \Lambda_{(t)} + \Delta_{\Lambda}$ and ignoring terms independent of $\Delta_\Lambda$,
	we consider 
	\begin{equation}
	\bar{h}_{(t)}(\Delta_\Lambda) = \vect(\nabla h_{(t)}(\Lambda_{(t)}))\trans\vect(\Delta_\Lambda) + \frac{1}{2}\vect(\Delta_\Lambda)\trans\nabla^2 h_{(t)}(\Lambda_{(t)}) \vect(\Delta_\Lambda),
	\notag
	\end{equation} 
	where 
	$\nabla h_{(t)}(\Lambda_{(t)})=\sy + 
	\bs_{(t)}\Theta_{(t)}^{T}\sx\Theta_{(t)}\bs_{(t)}
	+ 2\bs_{(t)}\Theta_{(t)}^T\sxy - \bs_{(t)}$ 
	and 
	$\nabla^2 h_{(t)}(\Lambda_{(t)})=\bs_{(t)} \otimes \bs_{(t)}$ 
	are gradient and Hessian matrices with respect to $\Lambda$ respectively,
	and $\otimes$ represents the Kronecker product.
	The Newton direction $D_{\Lambda,(t)}$ for \eqref{eq:Lambdaupdate} can  
	be written as the solution of the following regularized quadratic function 
	\cite{hsieh2011sparse}
	\begin{equation}
	D_{\Lambda,(t)} = \argminA_{\Delta_\Lambda}\Big\{\bar{h}_{(t)}(\Delta_\Lambda) + \lambda_{2} \norm{\Lambda_{(t)} + \Delta_\Lambda}_{1,\text{off}}\Big\}.
	\label{eq:2.2.1}
	\end{equation}
	Equation (\ref{eq:2.2.1}) can be solved efficiently 
	via the coordinate descent algorithm. Specifically, let $\Delta_{\Lambda,(0)} = 0$ 
	be the initial value, and $\Delta_{\Lambda,(s)}$ 
	be the update at iteration $s$. Then at iteration $s+1$, the $(i,j)$th element of 
	$\Delta_{\Lambda,(s)}$ is updated by
	\begin{equation}
	(\Delta_\Lambda)_{ij,(s+1)} \leftarrow (\Delta_\Lambda)_{ij,(s)} - 
	c_{\Lambda} + S_{\lambda_2 / a_{\Lambda}}\Big(c_{\Lambda} - 
	\frac{b_{\Lambda}}{a_{\Lambda}}\Big),
	\label{eq:2.2.1.1}
	\end{equation}
	where $a_{\Lambda} = \bs^2_{ij,(t)} + \bs_{ii,(t)}\bs_{jj,(t)}$, 
	$b_{\Lambda} = (\sy)_{ij} + \Big(\bs_{(t)}\Theta_{(t)}^{T}\sx\Theta_{(t)}\bs_{(t)}\Big)_{ij}
	+ 2\Big(\bs_{(t)}\Theta_{(t)}^T\sxy\Big)_{ij} - \bs_{ij,(t)} + (\bs_{(t)}\Delta_{\Lambda,(s)}\bs_{(t)})_{ij}$ 
	and $c_{\Lambda} = \Lambda_{ij,(t)} + (\Delta_{\Lambda})_{ij,(s)}$.
	Denote the penalized objective function at the $t$-th iteration as 
	$p_{(t)}(\Lambda) \triangleq h_{(t)}(\Lambda) + \lambda_2 \norm{\Lambda}_{1,\text{off}}$.
	We adopt the Armijo's rule \cite{armijo1966minimization} to find the step 
	size $\alpha$. Specifically, with a constant decrease rate $0 < \beta < 1$ 
	(typically $\beta = 0.5$), step sizes $\alpha = \beta^k$ for $k\in\mathbb{N}$
	are tried until the smallest $k$ such that 
	\begin{equation}
	p_{(t)}(\Lambda_{(t)}+\alpha D_{\Lambda,(t)}) \leq p_{(t)}(\Lambda_{(t)}) + 
	\alpha\sigma\Big\{ \tr(\nabla h_{(t)}(\Lambda_{(t)})D_{\Lambda,(t)}) + 
	\lambda_2 \norm{\Lambda_{(t)}+ D_{\Lambda,(t)}}_{1,\text{off}} - 
	\lambda_2 \norm{\Lambda_{(t)}}_{1,\text{off}}\Big\},
	\notag
	\end{equation}
	where $0 < \sigma < 0.5$ is the backtracking termination threshold.
	After the step size is calculated, we update 
	$\Lambda_{(t+1)} = \Lambda_{(t)} + \alpha D_{\Lambda,(t)}$. 
	
	When $n > \max(p, d)$, we use the maximum likelihood estimates 
	$\check{\Lambda}=(\sy - S_{xy}^{T} S_{xx}^{-1}\sxy)\inv$ and 
	$\check{\Theta}=-S_{xx}^{-1}\sxy\check{\Lambda}$ of $\bl$ and $\bt$ 
	as initial values for $\bl$ and $\bt$ respectively \cite{yin2011sparse}. 
	In the high dimensional case when $\sx$ is not invertible, we use
	the identity and zero matrix as initial values for $\bl$ and $\bt$ respectively. 
	
	The regularized Newton step (\ref{eq:2.2.1}) via the coordinate descent 
	algorithm described above is the most computational expansive part of the 
	algorithm. Despite its efficiency for lasso type of problems, updating all 
	$p(p+1)/2$ variables in $\Lambda$ is costly. To relieve this problem, 
	we divide the parameter set into an active set and a free set. As in 
	\citeasnoun{hsieh2011sparse} and \citeasnoun{wytock2013sparse}, at the $t$th  
	iteration of the algorithm, we only update $\Theta$ and $\Lambda$ over the 
	active set defined by
	\begin{equation}
	\begin{aligned}
	&\mathcal{S}_\Theta = \{ (i,j): \ |\big(\nabla_\Theta l_2(\Lambda_{(t)},\Theta_{(t)})\big)_{ij}| > \lambda_3 \  \text{or} \  \Theta_{{ij},(t)} \neq 0 \},\\
	&\mathcal{S}_\Lambda = \{ (i,j): \ |\big(\nabla h_{(t)}(\Lambda_{(t)})\big)_{ij}| > \lambda_2\  \text{or} \  \Lambda_{{ij},(t)} \neq 0 \}.
	\label{eq:2.2.6}
	\end{aligned}
	\end{equation}
	As the active set is relatively small due to sparsity induced by the 
	$\ell_1$ regularization, this strategy provides a substantial speedup.

	Tuning parameters $\lambda_2$ and $\lambda_3$ determine the sparsity of 
	$\Lambda$ and $\Theta$. As a general selection tool, leave-one-out or 
	$k$-fold cross-validation can be used to select these tuning parameters. 
	The leave-one-out cross-validation (LOOCV) can be computationally intensive 
	and various approximations have been proposed in the literature. 
	\citeasnoun{lian2011shrinkage} and 
	\citeasnoun{vujavcic2015computationally} derived generalized approximate 
	cross-validation (GACV) scores for selecting a single tuning parameter in 
	the GGM. The BIC and k-fold CV have been used to select a single tuning 
	parameter in the cGGM 
	\cite{yin2011sparse,sohn2012joint,wytock2013sparse,yuan2014partial,lee2012simultaneous}. 
	LOOCV has not been used for the cGGM as it requires fitting the model $n$ 
	times which is computationally intensive. To the best of our knowledge, there are no 
	computationally efficient alternatives to LOOCV in the current cGGM literature.
	We propose a new criterion, Leave-One-Out KL (LOOKL),
	for selecting $\lambda_2$ and $\lambda_3$ involved in \eqref{eq:pl2}
	as minimizers of 
	%\begin{equation}
	%\begin{aligned}
	%\text{LOOKL}(\lambda_2, \lambda_3) =& l_2(\hll, \htt) -
	%\sum_{k=1}^{n}\Big\{ \vecc(\Lambda\inv - \syk + 
	%\Lambda\inv\hat{\Theta}^T\sxk\htt\Lambda\inv)^T\\
	%&(-C + B^T A\inv B)\inv\big [ (F - B^T A\inv E)  \vecc(\sxkk - \sx)\\
	%& - (B^T A\inv H)\vecc(\sxykk - \sxy) + (G)\vecc(\sykk-\sy) \big ]\\ 
	%& + \vecc(-2\sxyk - 2\sxk\Theta\hll\inv)^T (-A+BC\inv B^T)\inv\\
	%& \big [ (E-BC\inv F) \vecc(\sxkk - \sx) -(BC\inv G)\vecc(\sykk-\sy)\\
	%& + (H)\vecc(\sxykk - \sxy) \big ] \Big\}.
	%\label{eq:LOOKL}
	%\end{aligned}
	%\end{equation}
	
	%\begin{eqnarray}
	%&&\text{LOOKL}(\lambda_2, \lambda_3) \nonumber \\
	%&=& l_2(\hll, \htt) -
	%\sum_{k=1}^{n}\Big\{ \vecc(\Lambda\inv - \syk + 
	%\Lambda\inv\hat{\Theta}^T\sxk\htt\Lambda\inv)^T
	%(-C + B^T A\inv B)\inv \nonumber \\
	%&&\big [ (F - B^T A\inv E)  \vecc(\sxkk - \sx)
	% - (B^T A\inv H)\vecc(\sxykk - \sxy) + (G)\vecc(\sykk-\sy) \big ] \nonumber \\ 
	%&& + \ \vecc(-2\sxyk - 2\sxk\Theta\hll\inv)^T (-A+BC\inv B^T)\inv \nonumber \\
	%&& \big [ (E-BC\inv F) \vecc(\sxkk - \sx) -(BC\inv G)\vecc(\sykk-\sy)
	% + (H)\vecc(\sxykk - \sxy) \big ] \Big\}.
	%\label{eq:LOOKL}
	%\end{eqnarray}
	
	\begin{eqnarray}
	&&\text{LOOKL}(\lambda_2, \lambda_3) \nonumber \\
	&=& -\frac{1}{n}l_2(\hll, \htt) + \frac{1}{2n}
	\sum_{k=1}^{n}\Big\{ \bu_k^T (-C + B^T A\inv B)\inv 
	\big [ (-E + B^T A\inv D)  \bv_{xx,k}
	+ 2B^T A\inv \bv_{xy,k} - \bv_{yy,k} \big ] \nonumber \\ 
	&& + \ \bw_k^T (-A+BC\inv B^T)\inv 
	\big [ (D-BC\inv E) \bv_{xx,k} + BC\inv \bv_{yy,k}
	-2 \bv_{xy,k} \big ] \Big\},
	\label{LOOKL}
	\end{eqnarray}
	where $\sxk = \bm{X}_k^T \bm{X}_k$, $\syk = \bm{Y}_k^T \bm{Y}_k$, 
	$\sxyk = \bm{Y}_k^T \bm{X}_k$, 
	$\sxkk = 1/n\sum_{i\neq k}S_{xx,i}$, $\sykk = 1/n\sum_{i\neq k}S_{yy,i}$, $\sxykk = 1/n\sum_{i\neq k}S_{xy,i}$,
	$\bu_k=\vecc(\hat{\Lambda}\inv - \syk + \hat{\Lambda}\inv\hat{\Theta}^T\sxk\htt\hat{\Lambda}\inv)$,
	$\bv_{xx,k}=\vecc(\sxkk - \sx)$,
	$\bv_{yy,k}=\vecc(\sykk-\sy)$,
	$\bv_{xy,k}=\vecc(\sxykk - \sxy)$,
	$\bw_k=\vecc(-2\sxyk - 2\sxk\hat{\Theta}\hll\inv)$,
	$A= -2 \hll\inv\otimes\sx$, 
	$B= 2\hll\inv\otimes\sx\htt\hll\inv$,
	$C=-\hll\inv\otimes(\hll\inv + 2\hll\inv\hat{\Theta}^T\sx\htt\hll\inv)$,
	$D= -2 \hll\inv\hat{\Theta}^T\otimes I_{d\times d}$, and
	$E = \hll\inv\hat{\Theta}^T\otimes\hll\inv\hat{\Theta}^T$.
	%$G= -I_{p^2\times p^2}$, and
	%$H= -2  I_{dp\times dp}$.
	The derivation is defered to \ref{appendix:LOOKL}. Note that the GACV in \citeasnoun{lian2011shrinkage} and KLCV in \citeasnoun{vujavcic2015computationally} are special cases of LOOKL with $\Theta = 0$. In the penalized case, we ignored the partial derivatives corresponding to the zero elements in $\Theta$ and $\Lambda$ \citeasnoun{lian2011shrinkage}, and showed that the LOOKL score remains the same. More details can be found in \citeasnoun{Luothesis}. Therefore, we conjecture that the proposed score is more appropriate for density estimation, rather than model selection.
   % {\color{red} corrected LOOKL definition}
	
	We note the proposed backfitting procedure and LOOKL method 
	for selecting tuning parameters are new for the cGMM.
	When $\Theta$ and $\Lambda$ are simultaneously updated using the second-order 
	Taylor expansion over all parameters \cite{wytock2013sparse}, an expensive 
	computation of the large Hessian matrix of size $(p+d)\times (p+d)$ is 
	required in each iteration. In contrast, our approach forms a second-order 
	approximation of a function of $\Lambda$ which requires a Hessian matrix of 
	size $p\times p$. The remaining set of parameters in $\Theta$ 
	can be updated easily using the simple coordinate descent algorithm. Moreover, 
	compared to the method in \citeasnoun{mccarter2016large}, our backfitting 
	algorithm eliminates the need for computing the large matrix 
	$\Sigma\Theta\trans\sx\Theta\Sigma$ in $\mathcal{O}(npd + np^2)$ time. Note 
	that we always require $\Lambda$ to be positive-definite after each iteration, 
	so the algorithm still has complexity $\mathcal{O}(p^3)$ flops due to the 
	Cholesky factorization.
	%In practice, the sparsity pattern to invert $\Lambda$ (???) can be exploited in much less than $\mathcal{O}(p^3)$ time by the sparse Cholesky decomposition \cite{mccarter2016large}. 

	Some off-the-shelf packages are utilized to solve the optimization problem. 
	Specifically, we use \texttt{QUIC} \cite{hsieh2014quic} for updating $\Lambda$, 
	and \texttt{gss} \cite{gu2014smoothing} for computing the smoothing spline 
	estimate of $f(\bx)$. We write R code for updating $\Theta$ using \eqref{eq:Thetaupdate}.
	We note that other penalities such as the smoothly clipped absolute deviation (SCAD)
	\cite{fan2001variable} and adaptive lasso \cite{zou2006adaptive} may be used to replace
	the $\ell_1$ penalty in the estimation of $\Theta$ and $\Lambda$. Details can be found 
	in \citeasnoun{Luothesis}.

	\section{Graph Estimation with cSScGG Models}
	\label{sec:graphestimation}
	
	In Section \ref{sec:densityestimation} we proposed the cSScGG model as a 
	flexible framework for estimating the multivariate density in 
	high-dimensional setting. In terms of the graph structure, the edges 
	among $\bm{Y}$ are identified by $\hat{\Lambda}$, and edges between 
	$\bm{X}$ and $\bm{Y}$ are identified by $\hat{\Theta}$ 
	\cite{sohn2012joint,wytock2013sparse,yuan2014partial}. The remaining 
	task is the identification of conditional independence within $\bm{X}$ 
	variables which is the target of this section.  
	
	We have assumed that the model space for the logistic density $\eta$
	contains a subset of components in the SS ANOVA decomposition 
	\eqref{eq:ssanovadecomp}. The interactions are often truncated to 
	overcome the curse of dimensionality and reduce the computational cost.
	As in \citeasnoun{gu2013smoothing} and \citeasnoun{gu2013nonparametric},
	in this section we consider the SS ANOVA model with all main effects and two-way 
	interactions as the model space for $\eta$. We note that the SS ANOVA 
	model allows pairwise nonparametric interactions as opposed to linear 
	interactions in the GGM.
	\citeasnoun{gu2013smoothing} and \citeasnoun{gu2013nonparametric} 
	proposed the squared error projection for accessing importance of each 
	interaction term and subsequently identify edges. However, we cannot 
	apply their method directly to $\hat{\eta}$ to identify edges within 
	$\bm{X}$ since the cGGM for $f(\by|\bx)$ also includes interaction 
	terms among variables in $\bm{X}$.
	
	The logarithm of the joint density 
	\begin{equation}
	\begin{aligned}
	\log f(\bz) &= \log f(\bx) + \log f(\by|\bx) \\
	&= \eta(\bx) + \frac{1}{2}\big(-\by\trans\bl\by - 
	2\bx\trans\bt\by - \bx\trans\bt\trans\bl\inv\bt\bx \big) + C,
	\notag
	\end{aligned}
	\end{equation}
	where $C$ is a constant independent of $\bx$ and $\by$. The main 
	challenge in identifying conditional independence among $\bm{X}$ 
	comes from the fact that $f(\by|\bx)$ brings in an extra term, 
	$- \bx\trans\bt\trans\bl\inv\bt\bx/2$, into the interactions among 
	$\bm{X}$. Let
	\begin{equation}
	\hat{\zeta}(\bx) = \hat{\Delta}(\bx) + \hat{\eta}(\bx).
	\end{equation}
	where $\hat{\Delta}(\bx) = -\bx^T\hat{\bt}\trans\hat{\bl}\inv\hat{\bt}\bx/2$. 
	Define the functional 
	\begin{equation}
	\tilde{V}(f, g) = \intx f(\bx)g(\bx)\rhox\dx - 
	\{ \intx f(\bx)\rhox\dx \}\{ \intx g(\bx)\rhox\dx \} 
	\label{tildeV}
	\end{equation}
	and denote $\tilde{V}(f, f)$ as $\tilde{V}(f)$. 
	Let $\mathcal{H} = \mathcal{S}^0\oplus\mathcal{S}^1$ where 
	$\mathcal{S}^1$ collects functions whose contribution to the 
	overall model is of question. The squared error projection of 
	$\zetah$ in $\mathcal{S}^0$ is \cite{gu2013smoothing}
	\begin{equation}
	\tilde{\zeta}=\argminA_{\zeta \in \mathcal{S}^0} 
	\big\{ \tilde{V}(\zetah - \zeta) \big\}.
	\label{eq:4.1.2}
	\end{equation}
	$\tilde{V}(\zetah - \zeta)$ can be regarded as a proxy of the symmetrized 
	KL divergence \cite{gu2013smoothing}. 
	Assuming $\zetau = - \log\rhox \in \mathcal{S}^0$, 
	it is easy to check that
	$\tilde{V}(\zetah-\zetau)=\tilde{V}(\zetah-\zetat)+\tilde{V}(\zetat-\zetau)$.
	Then the ratio $\tilde{V}(\zetah - \zetat)/\tilde{V}(\zetah - \zetau)$
	reflects the importance of functions in $\mathcal{S}^1$. 
	The quantity $\tilde{V}(\zetah - \zetau)$ is readily computable while 
	details for computing the squared error projection $\zetat$ are given in 
	\ref{appendix:projection}.
	
	For any pair of variables $X_i$ and $X_j$, consider the decomposition 
	$\mathcal{H}=\mathcal{S}^0_{ij}\oplus\mathcal{S}^1_{ij}$ where 
	$\mathcal{S}^1_{ij}$ is the subspace consisting of two-way interactions between $X_i$ and $X_j$, and $\mathcal{S}^0_{ij}$ contains all functions in $\mathcal{H}$ except the two-way interactions between $X_i$ and $X_j$. 
	Note that $\zeta_{ij}(x_i,x_j) \triangleq \eta_{ij}(x_i,x_j) + \hat{\Delta}_{ij}x_ix_j \in S_{ij}^1$ where
	 $\hat{\Delta}_{ij} = (\hat{\bt}\trans\hat{\bl}\inv\hat{\bt})_{ij}$.
       Compute the projection ratio
	$r_{ij} \triangleq \tilde{V}(\hat{\zeta} - 
	\tilde{\zeta})/\tilde{V}(\hat{\zeta} - \zeta_{u})$ in which 
	$\tilde{\zeta}$ is the squared error projection of $\hat{\zeta}$ in 
	$\mathcal{S}^0_{ij}$.
	The ratio
	$r_{ij}$ indicates the importance of interactions between $X_i$ and $X_j$, and we will add the interactions to the additive model sequentially according to the 
	descending order of $r_{ij}$'s.
	
	Consider the space decomposition 
	$\mathcal{H}=\mathcal{S}^0\oplus\mathcal{S}^1$.
	We start with  $\mathcal{S}^0$ being the subspace spanned by all main effects. We calculate the projection ratio of $\hat{\zeta}$ in 
	$\mathcal{S}^0$ as $r=\tilde{V}(\hat{\zeta} - \tilde{\zeta})/\tilde{V}(\hat{\zeta} - \zeta_u)$ 
	where $\tilde{\zeta}$ is the squared error 
	projection in $\mathcal{S}^0$. If $r$ is larger than a threshold, $\mathcal{S}^1$ is deemed important and we move the interaction with the largest $r_{ij}$ from $\mathcal{S}^1$ to $\mathcal{S}^0$. We then calculate the projection ratio $r$ with the updated $\mathcal{S}^0$ and $\mathcal{S}^1$. The projection ratio decreases each time we move an interaction from $\mathcal{S}^1$ 
	to $\mathcal{S}^0$. Finally the process stops when $r$ falls below a cut-off value, at which time we denote the corresponding $\mathcal{S}^0$ as $\mathcal{S}_s^0$. Let 
	$\Pi_{ij} = I\big(\zeta_{ij}\in\mathcal{S}_s^0\big)$ 
	and remove the edge between $X_i$ and $X_j$ if $\Pi_{ij} = 0$. In our implementations, the cut-off value is set to be $3\%$.

	To summarize, the conditional independences among $\bm{Y}$, between $\bm{X}$ 
	and $\bm{Y}$, and among $\bm{X}$ are characterized by the zero elements 
	in $\hat{\Lambda}$, $\hat{\Theta}$ and $\Pi$, respectively. 
	The whole procedure for edge identification is illustrated in Figure \ref{fig:4.1}.
	\begin{figure}[h!]
		\centering
		\includegraphics[width=0.6\textwidth]{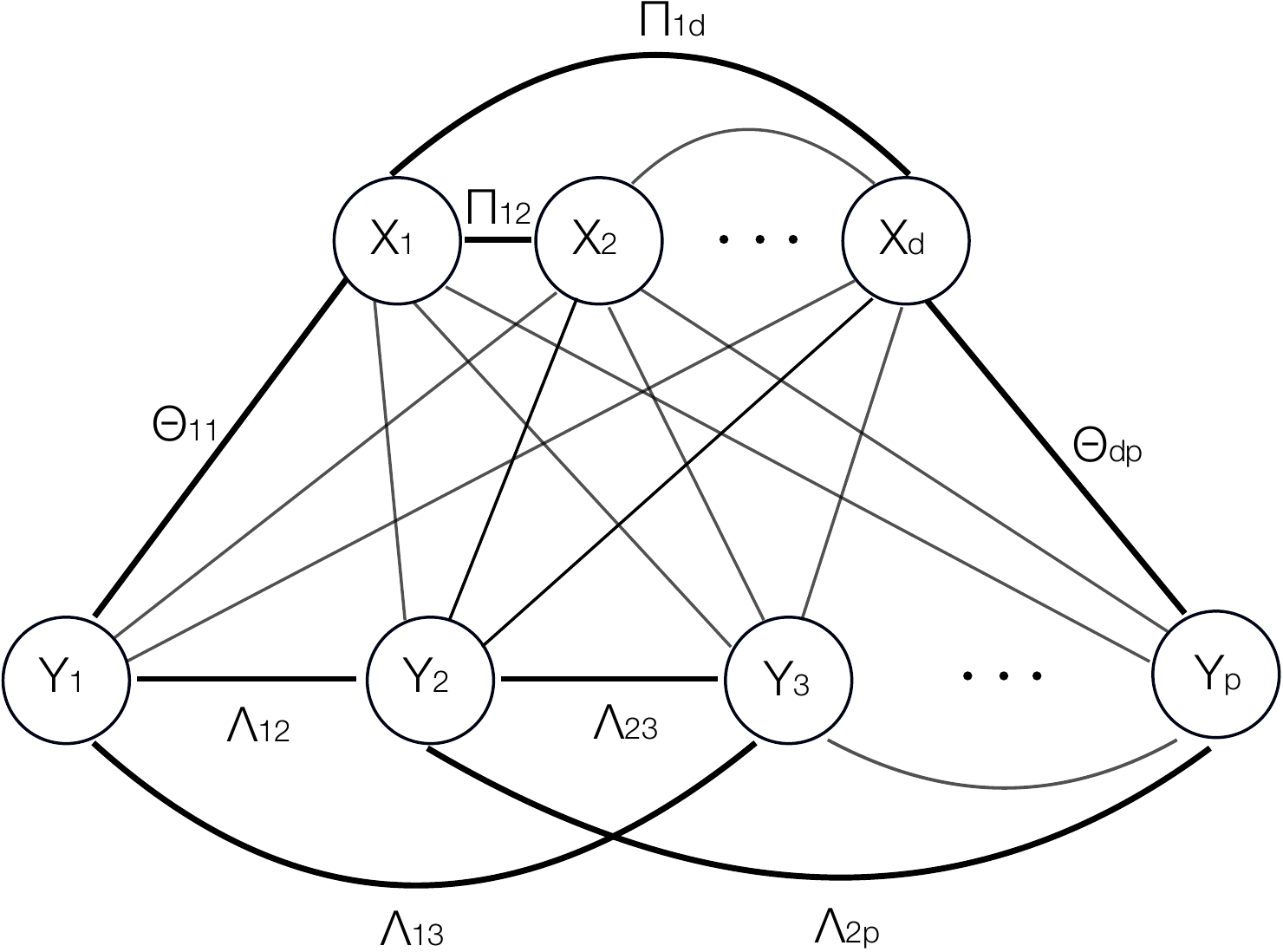}
		\caption{Illustration of the edge identification procedure.}
		\label{fig:4.1}
	\end{figure}

	\section{Theoretical Analysis}
	\label{sec:theory}
	
	We list notations, assumptions, and theoretical results only. Proofs are 
	given in \ref{appendix:proofs}.
	
	\subsection{Notations and Assumptions} 
	
	Given a matrix $U$, let $\vertiii{U}_2 =  \sqrt{\lambda_\text{max}(U^T  U)}$, 
	$\vertiii{U}_\infty=\max_{i=1,\dots,p}\sum_{j=1}^{p} |U_{ij}|$
	and 
	$\vertiii{U}_F=\sqrt{\sum_{i=1}^p\sum_{j=1}^p U_{ij}^2}$
	denote the $\ell_2$ operator norm , $\ell_\infty$ operator norm and Frobenius norm respectively, where $\lambda_\text{max}(U^T U)$ represents the largest eigenvalue of $U^T U$.
	We assume that $\bm{Y}|\bm{X} = \bm{x}\sim 
	\text{N}(-\Lambda_0^{-1}\Theta_0^{T} \bm{x},\Lambda_0^{-1})$ 
	where $\Lambda_0$ and $\Theta_0$ are the true parameters. 
	Let $\Gamma_0  = (\Lambda^{T}_0 , \Theta^{T}_0 )^T$,
	$\Sigma_0=\Lambda_0^{-1}$,
	$C_{\sigma} = \max_i\Sigma_{0,ii}$, 
	$C_{\Sigma}= \max_{i,j}|\Sigma_{0,ij}|$,
	$C_{\Theta} = \max_{i,j}|\Theta_{0,ij}|$,
	$\cx = \max_{j=1,\dots,d}\norm{\bm{X}^j}_2/\sqrt{n}$ 
	where $\bm{X}^j$ is the $j$th columns of $X$,
	$H_0 =\nabla^2_{\Lambda,\Theta} l_2(\Lambda_0,\Theta_0)$ denote the 
	Hessian matrix evaluated at the true parameters, and
	$\kappa_{H} = \max_{i,j}|H^{-1}_{0,ij}|$.
	Let $\gamma = \max_{1\leq j \leq p} \left\{ \sum_{i=1}^{d+p}
	I(\Gamma_{0,ij} \ne 0) \right\}$ 
	be the maximum number of non-zeros in any column of $\Gamma_0$
	which represents the maximum degree of $\bm{Y}$ in the graph.
	
	Denote $\lambda=\max \{\lambda_2,\lambda_3\}$ and 
	$r=\min \{\lambda_2,\lambda_3\}/\lambda$, then $r\leq 1$. 
	In the following theoretical analysis, we assume that $\lambda_2 \geq \lambda_3$ 
	and $r = \lambda_3/\lambda_2$. Similar arguments apply to the case of $\lambda_2 < \lambda_3$. The objective function can be rewritten as
	\begin{equation}
	\begin{aligned}
	&\{\hat{\Theta}, \hat{\Lambda} \} = \argminA_{\Lambda \succ 0,\Theta} 
	\big\{ l_2(\Lambda,\Theta) +\lambda(\norm{\Lambda}_1 + r\norm{  \Theta}_1 )\big\}.
	\end{aligned}
	\label{eq:obj}
	\end{equation}
	%where $f(\Lambda,\Theta) = -\log|\Lambda| + \tr(\sy\Lambda + 2S_{xy}^T \Theta + \Lambda^{-1}\Theta^T  S^T_{xx}\Theta)$.
	We make the following assumptions.
	
	\begin{assumption} (Underlying Model) 
		$\bm{Y}|\bm{X} = \bm{x}\sim 
		\text{N}(-\Lambda_0^{-1}\Theta_0^{T} \bm{x},\Lambda_0^{-1})$ 
		where $\bm{Y}$ has the maximum degree $\gamma$.
		\label{asmp:0}
	\end{assumption}
	
	\begin{assumption} (Restricted Convexity)
		For any $i=1,\dots,p$, 
		let $S_i$ denote the nonzero indices of the $i$-th column of 
		$\Theta_0$ (i.e., the edges between $\mathbf{X}$ and $Y_i$). We have
		$\lambda_{\emph{min}} (1/n X_{S_i}^T  X_{S_i} ) > 0$, 
		where $\lambda_\emph{min}(\cdot)$ denotes the smallest
		eigenvalue and $X_{S_i}$ represents the $n\times|S_i|$ matrix 
		with columns of $X$ indexed by $S_i$.
		\label{asmp:1}
	\end{assumption}
	
	\begin{assumption} (Mutual incoherence) 
		Let $S$ denote the support set of $\Gamma_0$ in vector form
		$S = (\emph{vec}(\emph{supp}\{\Lambda_0\})^T, \emph{vec}(\emph{supp}\{\Theta_0\})^T)^T$
		where $\emph{supp}\{\cdot\}$ denotes the indicator function 
		of whether an element is zero. 
		Let $\bar{S}$ denote the complement of $S$. We have
		$\vertiii{H_{0,\bar{S}S}(H_{0,SS})^{-1}}_\infty \leq 1-\alpha$
		for some $\alpha \in (0,1)$, where $H_{0,\bar{S}S}$ and $H_{0,SS}$ 
		represent the $|\bar{S}|\times |S|$ and $|S|\times |S|$ sub-matrices 
		of $H_0$ with entries in $\bar{S} \times S$ and $S \times S$ respectively. 
		\label{asmp:2}
	\end{assumption}
	
	\begin{assumption} (Control of eigenvalues) There exists some 
		constants $0 < C_L \le C_U < \infty$, such that
		$C_L \le \lambda_\mathrm{min}(\Lambda_0) \leq \lambda_\mathrm{max}(\Lambda_0) \le C_U$.
		\label{asmp3}
	\end{assumption}

	Assumption \ref{asmp:0} provides the true underlying model. Assumption \ref{asmp:1} ensures the solution of optimization problem (\ref{eq:obj})
	is restricted to the active set (nonzero entries in 
	$\Lambda_0$ and $\Theta_0$), which is also used in 
	\citeasnoun{wainwright2009sharp} and \citeasnoun{wytock2013sparse}.
	Assumption \ref{asmp:2} limits the influence of edges in 
	inactive set ($\bar{S}$) can have on the edges in active set ($S$), and Assumption \ref{asmp3} bounds the eigenvalues of the precision matrix. 
	Define $V(f,g)=\intx f(\bm{x}) g(\bm{x}) \rhox d\bm{x}$
	and $V(f)=\intx f^2(\bm{x})\rhox d\bm{x}$.
	
	\begin{assumption} $V$ is completely continuous with respect to $J$
		and $J(\eta_0)<\infty$.
		\label{asmp4}
	\end{assumption}
	
	Under the Assumption \ref{asmp4}, there exists $\phi_\nu $ such that 
	$V(\phi_\nu,\phi_\mu) = \delta_{\nu,\mu}$, 
	$J(\phi_\nu,\phi_\mu) = \rho_\nu\delta_{\nu,\mu}$, and 
	$0\leq \rho_\nu \uparrow \infty$, where $\delta_{\nu,\mu}$ is the 
	Kronecker delta and $\rho_\nu$ is referred to as the eigenvalues 
	of $J$ with respect to $V$. 
	Denote the Fourier series expansion of $\eta_0$ as 
	$\eta_0=\sum_{\nu}\eta_{\nu,0}\phi_\nu$ where
	$\eta_{\nu,0} = V(\eta_0, \phi_\nu)$ are the Fourier coefficients.
	Let $\tilde{\eta}=\sum_{\nu}\tilde{\eta}_{\nu}\phi_\nu$ where 
	$\tilde{\eta}_\nu = (\beta_\nu + \eta_{\nu,0})/(1+\lambda_1\rho_\nu)$
	and 
	$\beta_\nu = n^{-1}\sum_{i=1}^{n}\{e^{-\eta_0(\bm{X}_i)}
	\phi_\nu(\bm{X}_i)- \intx \phi_\nu(\bm{x}) \rhox d\bm{x}\}$.

	\begin{assumption} \label{asmp5}
		(a)
		The eigenvalues $\rho_\nu$ of $J$ with respect to $V$ 
		satisfy $\rho_\nu > \beta \nu^s$ for some $\beta > 0$ and $s > 1$ 
		when $\nu$ is sufficiently large.
		
		\noindent
		(b) There exists some constants 
		$0 < C_{1,1} < C_{1,2} < \infty$, $C_{1,3} < \infty$ and $C_{1,4} < \infty$ 
		such that
		$C_{1,1} < e^{\eta_0(\bm{x})-\eta(\bm{x})} < C_{1,2}$ holds uniformly for 
		$\eta$ in a convex set around $\eta_0$ containing $\etah$ and $\etat$,
		$e^{-\eta_0(\bm{x})} < C_{1,3}$, and 
		$\intx \phi^2_\nu(\bm{x})\phi^2_\mu(\bm{x})e^{-\eta_0(\bm{x})}\rhox d\bm{x}< C_{1,4}$
		for any $\nu$ and $\mu$.
		
		\noindent
		(c) There exists some $q\in [1,2]$ such that 
		$\sum_{\nu}\rho_\nu^q\eta_{\nu,0}^2 < \infty$.
		\label{asmp5}
	\end{assumption}
	Assumptions \ref{asmp4} and \ref{asmp5} are commonly used in smoothing spline 
	literatures to study the convergence rate for 
	nonparametric density estimation \cite{gu2013smoothing} .

	\subsection{Asymptotic Consistency of the Estimated Parameters}
	
	The following theorem provides the estimation error bound and edge 
	selection accuracy.
	\begin{theorem}
		Suppose that the Assumptions \ref{asmp:1} and \ref{asmp:2} hold, $\tau > 2$, and 
		$n$ and $\lambda=\max\{\lambda_2,\lambda_3 \}$ satisfy 
		\begin{eqnarray}
		n &\geq& C_{2,1} C_{2,2}^2 C_{\sigma}^2 \gamma^4 
		( 1+8\alpha^{-1})^4 [\tau\log(pd) + \log 4 ], \label{eq:thm1.1} \\
		\lambda &=& 8 \alpha^{-1} C_{\sigma}C_X^\star\sqrt{3200} 
		\sqrt{\frac{\tau\log (pd) + \log 4}{n}}, \notag
		\end{eqnarray}
		where $C_{2,1} = \max\{ 12800, 32C_X^2 \}$, 
		$C_{2,2} = \kappa_{H} \max\{ 3C_{\Sigma}/\gamma, 
		2/(C_{\Theta}\gamma), 412C_{\Sigma}^4C_{\Theta}^2C_X^2 \}$, 
		and $C_X^\star = \max\{C_X^2,1\}$, then with probability greater than 
		$1-\big(p^{-(\tau-2)}+ (pd)^{-(\tau-1)}\big)$, we have
		\begin{enumerate}
			\item The estimates satisfy the elementwise $\ell_\infty$ bound:
			\begin{equation}
			\max\Big\{ \norm{\hat{\Lambda}-\Lambda_0}_\infty, \norm{\hat{\Theta}-\Theta_0}_\infty \Big\} \leq 2\kappa_{H} (1 + 8\alpha^{-1}) C_{\sigma}C_X^\star\sqrt{3200} \sqrt{\frac{\tau\log (pd) + \log 4}{n}}.\notag
			\end{equation}
			\item All non-zero entries of the solution $(\hat{\Lambda}, \hat{\Theta})$ are a subset of the non-zero entries of $(\Lambda_0, \Theta_0)$. Furthermore, 
			non-zero entries of $(\hat{\Lambda}, \hat{\Theta})$ includes all non-zero entries $(i,j)$ in $(\Lambda_0, \Theta_0)$ that satisfy
			\begin{equation}
			\min\{ \Lambda_{0,ij}, \Theta_{0,ij} \} >  4 \kappa_{H}(1+8\alpha^{-1}) C_{\sigma}C_X^\star\sqrt{3200} \sqrt{\frac{\tau\log (pd) + \log 4}{n}}.
			\label{eq:5.2.3}
			\end{equation} 
		\end{enumerate}
		\label{thm:1}
	\end{theorem}
	
	\noindent\textbf{Remark 1:} i) Theorem \ref{thm:1} indicates that a sample 
	size larger than a constant times $\gamma^4 \log(pd)$ is enough for our
	estimation procedure to identify a subset of the true non-zero 
	elements in the cGGM, and the resulting estimations are close to the 
	true parameters in $\ell_\infty$ bound. The convergence rate is the same as 
	that in \citeasnoun{wytock2013sparse}, but the success probability of 
	the primal-dual witness approach as well as the exact bounds for $n$ 
	and $\lambda$ are different. We also provide a lower bound for the 
	sign consistency which is not included in \citeasnoun{wytock2013sparse}.
	
	\noindent ii) The convergence probability is smaller than that for the GGM 
	\cite{wainwright2009sharp} where only a precision matrix needs to be 
	estimated. This is the price we pay for estimating extra parameters in $\Theta$.
	
	Define $s_\Lambda$ as the total number 
	of non-zero elements in off-diagonal positions of $\Lambda_0$, and
	$s_\Theta$ as the total number of non-zero elements in $\Theta_0$. 
	
	\begin{corollary}
		Under the same assumptions as in Theorem 1, with probability at least $1-\big(p^{-(\tau-2)}+ (pd)^{-(\tau-1)}\big)$, the estimates $\hat{\Lambda}$ and $\hat{\Theta}$ satisfy
		\begin{eqnarray}
		&&\max\Big\{ \vertiii{\hat{\Lambda}-\Lambda_0}_F, 
		\vertiii{\hat{\Theta}-\Theta_0}_F \Big\} \nonumber \\
		&\leq& 2\kappa_{H} (1 + 8\alpha^{-1}) \max\{\sqrt{p+s_\Lambda},\sqrt{s_\Theta}\}
		C_{\sigma}C_X^\star\sqrt{3200} \sqrt{\frac{\tau\log (pd) + \log 4}{n}}.
		\label{eq:col1}
		\end{eqnarray}
		\label{col:1}
	\end{corollary}
	
	\noindent\textbf{Remark 2:} The Frobenius norm was not studied in 
	\citeasnoun{wytock2013sparse}. We develop it as a building block for 
	establishing the convergence rate for the density estimation in Section 
	\ref{sec:densityrate}.
	
	%Proofs of Theorem \ref{thm:1} and Corollary 1 are given in Appendix ?
	
	\subsection{Convergence Rates for the Density Estimation}
	\label{sec:densityrate}

	We first introduce a combined measure of divergence between the 
	joint density and its estimate. 
	Let $f_0(\bm{x}) = e^{\eta_0(\bm{x})}\rho(\bm{x})/ 
	\intx e^{\eta_0(\bm{x})}\rho(\bm{x}) d\bm{x}$ and $f_0(\by|\bm{x})$
	be the true densities of $\bm{X}$ and $\bm{Y}|\bm{X}=\bm{x}$ 
	with their estimates denoted as 
	$\hat{f}(\bm{x}) = e^{\hat{\eta}(\bm{x})}\rhox/
	\intx e^{\hat{\eta}(\bm{x})}\rhox d\bm{x}$
	and $\hat{f}(\by|\bm{x})$ respectively. The KL
	divergence between two density functions $f_1$ and $f_2$ are defined as 
	$\text{KL}(f_1,f_2)=\text{E}_{f_1}[\log(f_1/f_2)]$. Then 
	the symmetrized KL divergence between the true joint density
	$f_0(\bz)=f_0(\bm{x})f_0(\by|\bm{x})$ and its estimate 
	$\hat{f}(\bz)=\hat{f}(\bm{x})\hat{f}(\by|\bm{x})$ can expressed as
	\begin{eqnarray}
	&&\text{SKL}\Big(f_0(\bz),\hat{f}(\bz)\Big) \nonumber\\
	&=& \text{KL}\Big(f_0(\bz),\hat{f}(\bz)\Big) + 
	\text{KL}\Big(\hat{f}(\bz),f_0(\bz)\Big) \nonumber \\
	&=& \left\{ \int_\mathcal{X} f_0(\bm{x}) \text{KL}\Big(f_0(\by|\bm{X}=\bx), \hat{f}(\by|\bm{X}=\bx)\Big)d\bx + \int_\mathcal{X} \hat{f}(\bm{x}) \text{KL}\Big( \hat{f}(\by|\bm{X}=\bx),f_0(\by|\bm{X}=\bx)\Big)d\bx \right\} \nonumber \\
	&+& \left\{ \text{KL}\Big(f_0(\bm{x}), \hat{f}(\bm{x})\Big) + 
	\text{KL}\Big( \hat{f}(\bm{x}),f_0(\bm{x})\Big) \right\} \nonumber \\
	&\triangleq& \text{SKL}\big(f_0(\by|\bx),\hat{f}(\by|\bx) \big) + 
	\text{SKL}\big(f_0(\bm{x}),\hat{f}(\bm{x})\big).
	\label{eq:5.2.1}
	\end{eqnarray}
	
	We will establish the asymptotic convergence 
	rate under the following combined measure of divergence
	\begin{equation}
	D\big(f_0(\bz),\hat{f}(\bz)\big) \triangleq   
	\text{SKL}\big(f_0(\by|\bx),\hat{f}(\by|\bx)\big) + (V+\lambda_1 J)(\eta_0-\etah).
	\label{eq:5.2.2}
	\end{equation}
	The difference between (\ref{eq:5.2.2}) and (\ref{eq:5.2.1}) lies in 
	the divergence measures for the estimation of $f(\bm{x})$. 
	Note that the rate in $V(\eta-\eta_0)$ implies rate in 
	$\tilde{V}(\eta-\eta_0)$, since $\tilde{V}(f) \le V(f)$ and $\tilde{V}(\eta_0-\hat{\eta})$ is a proxy of 
	$\text{SKL}(f_0(\bm{x}),\hat{f}(\bm{x}))$ \cite{gu2013smoothing}. 
	
	We first establish the rate for 
	$\text{SKL}(f_0(\by|\bx),\hat{f}(\by|\bx))$ which has the 
	explicit expression:
	\begin{eqnarray}
	&& \text{SKL} \big( f_0(\by|\bx),\hat{f}(\by|\bx) \big) \nonumber\\
	&=& \frac{1}{2}\bm{a}^T\hat{\Lambda}\bm{a}\int_\mathcal{X} \bm{x}^T\bm{x}f_0(\bx)d\bx +\nonumber \frac{1}{2}\bm{a}^T\Lambda_0\bm{a}\int_\mathcal{X} \bm{x}^T\bm{x}\hat{f}(\bx)d\bx 
	+ \frac{1}{2}\tr\big( \hat{\Lambda}^{-1}\Lambda_0 + 
	\Lambda_0^{-1}\hat{\Lambda} \big) - p,
	\label{eq:thm2.2}
	\end{eqnarray}
	where $\mathbf{a} = \hat{\Lambda}^{-1}\hat{\Theta}^T  - 
	\Lambda_0^{-1}\Theta_0^T$. We assume that the second moments of marginal 
	densities of $f_0$ and $\hat{f}$ exist. 
	
	\begin{theorem}
		Under the Assumption \ref{asmp3} and conditional on the event 
		$\vertiii{\hat{\Lambda}-\Lambda_0}_F \leq 0.5 C_L$, we have 
		\begin{equation}
		\text{SKL}\big(f_0(\by|\bx),\hat{f}(\by|\bx)\big) = \mathcal{O}\Big(n^{-5/2}p^{5/2}(\log pd)^{5/2} + n^{-1}p^2(\log pd)\Big).
		\label{eq:thm2.1}
		\end{equation}
		\label{thm:2}
	\end{theorem}
	
	For the smoothing spline ANOVA estimate $\etah$ of $\eta_0$, 
	under the Assumptions \ref{asmp4} and \ref{asmp5}, \citeasnoun{gu2013smoothing} showed that 
	as $\lambda_1\rightarrow 0$ and $n\lambda_1^{2/s}\rightarrow \infty$,
	\begin{equation}
	(V + \lambda_1 J)(\etah-\eta_0) = \mathcal{O}(n^{-1}\lambda_1^{-1/s}+\lambda_1^q).
	\label{eq:thm3.2}
	\end{equation}
	
	Finally, we have the convergence rate for the joint density estimate.
	
	\begin{theorem}
		Suppose that the Assumptions \ref{asmp:1}-\ref{asmp5} hold, $\tau > 2$, 
		$\lambda_1\rightarrow 0$, $n\lambda_1^{2/s}\rightarrow \infty$, and 
		$n$ and $\lambda$ satisfy
		\begin{eqnarray}
		n &\geq& C_{3,1} C_{3,2}^2 C_{\sigma}^2 \max\{\gamma^4, p+s_\Lambda\} 
		(1+8\alpha^{-1})^4 [\tau\log(pd) + \log 4 ], \label{thm:3.1} \\
		\lambda &=& 8\alpha^{-1} C_{\sigma}C_X^\star\sqrt{3200} 
		\sqrt{\frac{\tau\log (pd) + \log 4}{n}},\notag
		\end{eqnarray}
		where $C_{3,1} = C_{2,1} = \max\{ 12800, 32C_X^2 \}$, and
		$C_{3,2} = \max\{C_{2,2}, \kappa_{H}\sqrt{1600}/C_L \}=\kappa_{H} \max\{ 3C_{\Sigma}/\gamma, 2/(C_{\Theta}\gamma), \\
		412C_{\Sigma}^4C_{\Theta}^2C_X^2, \sqrt{1600}/C_L \}$, then
		with probability greater than $1-\big(p^{-(\tau-2)}+ (pd)^{-(\tau-1)}\big)$ we have
		\begin{equation}
		D\big(f_0(\bz),\hat{f}(\bz)\big) =  \mathcal{O}\big( n^{-5/2}p^{5/2}(\log pd)^{5/2} + n^{-1}p^2(\log pd) + n^{-1}\lambda_1^{-1/s}+\lambda_1^q \big).
		\label{thm:3.2}
		\end{equation}
		\label{thm:3}
	\end{theorem}
	
	\noindent
	\textbf{Remark 3:} For low-dimensional $\bm{X}$ (usually $d\leq 3$), 
	the computation of multivariate integrals are feasible. We may use the 
	penalized likelihood instead of the pseudo likelihood to estimate the 
	density function $f(\bx)$. This leads to 
	$f_0(\bx) = e^{\eta_0}/\intx e^{\eta_0}$. 
	Under similar conditions, \citeasnoun{gu2013smoothing} has proved that 
	the symmetrized KL divergence $\text{SKL}(f_0(\bm{x}),\hat{f}(\bm{x}))$ 
	is also $\mathcal{O}(n^{-1}\lambda_1^{-1/s}+\lambda_1^q)$,
	where $\hat{f}(\bm{x})$ is the penalized likelihood estimate. 
	If we also use the penalized likelihood to estimate $\eta$ in our model, then $\text{SKL}\big(f_0(\bz),\hat{f}(\bz)\big)= \mathcal{O}(n^{-5/2}p^{5/2}(\log pd)^{5/2} + n^{-1}p^2(\log pd) + n^{-1}\lambda_1^{-1/s}+\lambda_1^q)$.

\section{Simulation Studies}
\label{sec:simulation}
We have conducted extensive simulation experiments to evaluate the performance of the cSScGG procedure, and compare it with some existing parametric and semiparametric/nonparametric methods. To save space, we present some simulation results and more comprehensive results can be found in \citeasnoun{Luothesis}.
We note that the cSScGG method can ourperform the maximum likelihood estimation (MLE) when $\bm{Z} = (\bm{X}^T, \bm{Y}^T)^T$ is multivariate Gaussian and the cGGM for $\bm{Y}$ is sparse. 
Results for density and graph estimations are presented in Sections \ref{sec:sim_1} and \ref{sec:sim_2} respectively.

For density estimation, we use both LOOKL and CV (5-fold) methods to choose $\lambda_2$ and $\lambda_3$. Tuning parameters involved in all other methods are chosen by 5-fold CV. For graph estimation, we select $\lambda_2$ and $\lambda_3$ in the cSScGG method as minimizers of the
following BIC score
\begin{equation}
\text{BIC}(\lambda_2,\lambda_3) =  \Big\{-n \log|\hat{\Lambda}| + n\tr(\sy\hat{\Lambda} + 2 S_{xy}^T\hat{\Theta} + \hat{\Lambda}\inv\hat{\Theta}^T\sx\hat{\Theta}) \Big \} + \log n \{ \xi(\hat{\Lambda}) / 2 + \xi(\hat{\bt}) \},
\end{equation}
where $\xi(\hat{\Lambda})$ and $\xi(\hat{\bt})$ are the number of non-zero off-diagonal elements in $\hat{\Lambda}$ and the number of non-zero elements in $\hat{\bt}$ respectively. 
The degree of freedom is defined in the same way as in \citeasnoun{yin2011sparse}. The BIC is also used to select tuning parameters in 
other methods for graph estimation.
More details regarding comparison of various tuning parameter selection methods are included in \citeasnoun{Luothesis}.

\subsection{Density Estimation}
\label{sec:sim_1}

We set $n=200$, $d=3$, and $p=25$. We generate $\bm{X}\sim \omega \mathcal{N}(\bm{\mu}_1,\sigma^2I) + (1-\omega) \mathcal{N}(\bm{\mu}_2,\sigma^2 I)$ with $\bm{\mu}_1 = (1,0,-1)^T$, and $\bm{\mu}_2 = (0,-1,1)^T$. We consider four combinations of $\sigma$ and $\omega$: $\sigma = 0.5, 0.1$ and $\omega = 0.9, 0.1$. All results are reported based on 100 replications under each setting. In each replication, we first generate $n$ iid samples $\bm{X}_1,\dots,\bm{X}_n$ from the multivariate Gaussian mixtures, then $\bm{Y}_i$'s are generated from a cGGM. Specifically, we randomly create a $(d+p)\times(d+p)$ precision matrix $\Omega$ using the R-package \texttt{huge} \cite{zhao2012huge}, in which the probability of the off-diagonal elements being nonzero equals $0.2$. The decomposition $\Omega=\begin{bmatrix}
\Omega_{xx} & \Omega_{xy} \\
\Omega_{yx} & \Omega_{yy}
\end{bmatrix}$ gives us $\Theta = \Omega_{xy}$ and $\Lambda = \Omega_{yy}$ \cite{yuan2014partial}, so that we can sample $\bm{Y}_i$ from $\mathcal{N}(-\Lambda\inv\Theta^T\bm{X}_i,\Lambda\inv)$ for $i=1,\dots,n$.

Since the division of non-Gaussian variables $\bm{X}$ and Gaussian variables
$\bm{Y}$ is typically unknown in practice, we consider two versions of the proposed method -- plain cSScGG and cSScGG with normality test (denoted as NT). In the plain version, we assume that the true non-Gaussian components are known and apply cSScGG directly. In the NT version, we select $d$ variables with smallest p-values based on the Shapiro-Wilk test to all $p+d$ marginal variables as
$\bm{X}$, and then apply the cSScGG method. 

In addition to the cSScGG method, we estimate density using the SKDE \cite{hoti2004semiparametric}, MLE, and QUIC \cite{hsieh2011sparse} methods. In the implementation of the SKDE method, we use the R-package \texttt{ks} \cite{duong2007ks} to calculate the kernel density estimate for $f(\bx)$ with the bandwidth selected by the smoothed cross-validation selector with diagonal bandwidth matrices (\texttt{Hscv.diag(x)}) which provides the best overall performance. To avoid selecting the two extra bandwidths involved in SKDE, as in \citeasnoun{hoti2004semiparametric}, we set $f(\by|\bx) = f(\by)$ and use MLE to estimate $f(\by)$. MLE and QUIC methods treat $\bm{Z}^T = (\bm{X}^T,\bm{Y}^T)$ as multivariate normal across all settings, and the estimates from these two methods are further broken down into $f(\bx)$ and $f(\by|\bx)$ for comparison. Specifically, QUIC method learns the precision matrix of $\bm{Z}$ by forming a quadratic approximation of the log-likelihood, and the estimates are computed using the R-package \texttt{QUIC} \cite{hsieh2014quic}. \\
To evaluate the performance of different methods, we consider the KL divergence between the estimated density and the true density
\begin{equation}
			\begin{aligned}
			\text{KL}\Big(f_0(\bz),\hat{f}(\bz)\Big) = \text{E}_{\bm{X}} \Big[\text{KL}\Big(f_0(\by|\bm{X}), \hat{f}(\by|\bm{X})\Big)\Big] + \text{KL}\Big(f_0(\bx), \hat{f}(\bx)\Big),
			\end{aligned}
			\notag
\end{equation}
where $f_0$ is the true density, and the aggregated KL $\text{E}_{\bm{X}} \Big[\text{KL}\Big(f_0(\by|\bm{X}), \hat{f}(\by|\bm{X})\Big)\Big]$ is approximated by the empirical aggregated KL divergence. Table \ref{tb:2.8} reports the overall KL divergence $\text{KL}\Big(f_0(\bz),\hat{f}(\bz)\Big)$, the empirical aggregated KL divergence $n^{-1}\sum_{i=1}^{n} \text{KL}\Big(f_0(\by|\bm{X}_i), \hat{f}(\by|\bm{X}_i)\Big)$, and $\text{KL}\Big(f_0(\bx), \hat{f}(\bx)\Big)$. They provide evaluations for the estimation of $f(\bz)$, $f(\by|\bx)$, and $f(\bx)$, respectively. \\

\begin{table}[H]
	\centering
	\scalebox{0.85}{
		\setlength{\tabcolsep}{0.5em} % for the horizontal padding
		{\renewcommand{\arraystretch}{1}% for the vertical padding
			\begin{tabular}{|l|l|llll|}
				\hline
		KL		& Method    & \multicolumn{1}{c}{\begin{tabular}[c]{@{}c@{}}$\sigma=0.5$\\ $\omega=0.9$ \end{tabular}} &  \multicolumn{1}{c}{\begin{tabular}[c]{@{}c@{}}$\sigma=0.5$\\ $\omega=0.5$ \end{tabular}} & \multicolumn{1}{c}{\begin{tabular}[c]{@{}c@{}}$\sigma=0.1$\\ $\omega=0.9$\end{tabular}} &  \multicolumn{1}{c|}{\begin{tabular}[c]{@{}c@{}}$\sigma=0.1$\\ $\omega=0.5$\end{tabular}} \\ \hline
				\multicolumn{1}{|l|}{\multirow{4}{*}{$f(\bx)$}} &cSScGG & 0.030 (0.023) & 0.043 (0.021) & 0.062 (0.073) & 0.046 (0.026) \\ 
				\multicolumn{1}{|l|}{}                        &SKDE & 0.208 (0.172) & 0.181 (0.190) & 0.503 (0.414) & 0.141 (0.091) \\ 
				\multicolumn{1}{|l|}{}                        &QUIC & 0.225 (0.023) & 0.243 (0.012) & 3.313 (0.051) & 3.528 (0.021) \\ 
				\multicolumn{1}{|l|}{}                        &MLE & 0.182 (0.021) & 0.225 (0.010) & 3.128 (0.047) & 3.487 (0.023) \\ 
				\hline
				\multicolumn{1}{|l|}{\multirow{4}{*}{$f(\by|\bx)$}} & cSScGG\_CV & 1.145 (0.202) & 1.118 (0.189) & 1.040 (0.138) & 1.078 (0.175) \\
				\multicolumn{1}{|l|}{}                        &cSScGG\_LOOKL & 1.143 (0.164) & 1.098 (0.167) & 1.14 (0.172) & 1.112 (0.161) \\
				\multicolumn{1}{|l|}{}                        &SKDE & 1.621 (0.184) & 1.632 (0.218) & 1.425 (0.173) & 1.474 (0.215) \\ 
				\multicolumn{1}{|l|}{}                        &QUIC & 1.196 (0.125) & 1.163 (0.147) & 1.179 (0.141) & 1.156 (0.139) \\ 
				\multicolumn{1}{|l|}{}                        &MLE & 1.827 (0.245) & 1.613 (0.236) & 2.235 (0.413) & 1.607 (0.235) \\ 
				\hline
				\multicolumn{1}{|l|}{\multirow{4}{*}{$f(\bz)$}}   &cSScGG\_CV & 1.175 (0.205) & 1.161 (0.189) & 1.102 (0.155) & 1.124 (0.178) \\ 
				\multicolumn{1}{|l|}{}   & cSScGG\_CV\_NT & 1.262 (0.208) & 1.268 (0.173) & 1.032 (0.125) & 1.051 (0.120) \\
				\multicolumn{1}{|l|}{}   & cSScGG\_LOOKL & 1.173 (0.167) & 1.141 (0.166) & 1.202 (0.186) & 1.158 (0.164) \\
				\multicolumn{1}{|l|}{}   & cSScGG\_LOOKL\_NT & 1.358 (0.241) & 1.325 (0.211) & 1.153 (0.184) & 1.122 (0.166) \\
				\multicolumn{1}{|l|}{}                        &SKDE & 1.829 (0.256) & 1.813 (0.331) & 1.928 (0.455) & 1.615 (0.234) \\ 
				\multicolumn{1}{|l|}{}                        &QUIC & 1.422 (0.132) & 1.405 (0.148) & 4.492 (0.15) & 4.684 (0.137) \\ 
				\multicolumn{1}{|l|}{}                        &MLE & 2.009 (0.245) & 1.838 (0.237) & 5.363 (0.404) & 5.094 (0.232) \\
				\hline
			\end{tabular}
		}}
		\caption{Averages and standard deviations (in parentheses) of the overall KL divergence $\text{KL}\Big(f_0(\bz),\hat{f}(\bz)\Big)$ (denoted by $f(\bz)$), the empirical aggregated KL $1/n\sum_{i=1}^{n} \text{KL}\Big(f_0(\by|\bm{X}_i), \hat{f}(\by|\bm{X}_i)\Big) $ (denoted by $f(\by|\bx)$), and $\text{KL}\Big(f_0(\bx), \hat{f}(\bx)\Big)$ (denoted by $f(\bx)$). cSScGG\_CV (cSScGG\_LOOKL) and cSScGG\_CV\_NT (cSScGG\_LOOKL\_NT) correspond to the cSScGG method without and with normality test repectively, and tuning parameters $\lambda_2$ and $\lambda_3$ are selected by the 5-fold CV (LOOKL).}
		\label{tb:2.8}
	\end{table}
Since cSScGG with normality test may identify different $\bm{X}$, we only include the overall KL divergence $\text{KL}\Big(f_0(\bz),\hat{f}(\bz)\Big)$ for comparison. Both versions of the cSScGG method enjoy superior performance relative to all other methods under all settings.
%This is not surprising as cSScGG with normality test is data-driven. It selects components that behave most differently from Gaussian, and may result in better density estimation.
When comparing the plain cSScGG with other methods, the differences mainly come from the estimation of $f(\bx)$, in which parametric methods MLE and QUIC cannot fit the data properly. The cSScGG performs much better than SKDE in both the estimation of $f(\bx)$ and $f(\by|\bx)$. When $\sigma$ is fixed, the performance differences are larger under $\omega=0.5$ where the deviation from Gaussian is more severe. Furthermore, under a fixed $\omega$, the superiority of the cSScGG methods is greater when $\sigma=0.1$ where the deviation from Gaussian is more severe. Comparative results remain the same under other simulation settings \cite{Luothesis}.

\subsection{Edge Detection}
\label{sec:sim_2}
We do not consider the SKDE and MLE methods here because we they do not perform edge selection. In addition to QUIC which is parametric, we also include the nonparanormal (NPN) method \cite{liu2009nonparanormal}. The NPN method is implemented with the R-package \texttt{huge}. When fitting the model, we use shrunken ECDF to transform the data first, then apply Glasso to the transformed data. The final NPN model is selected by the extended BIC score \cite{foygel2010extended}. Given a fixed dimension $p$, the model chosen by the EBIC method agrees with the model chosen by the BIC method. As the cSScGG method is formulated quite differently from the NPN, our main focus is to investigate the improvements that cSScGG can bring over the QUIC method which assumes normality for all variables including $\bm{X}$.\\
The performance is measured in three categories: among $\bm{X}$, among $\bm{Y}$, and between $\bm{X}$ and $\bm{Y}$. Recall that for the cSScGG procedure, edges in the above categories are decided by $\Pi$, $\Lambda$ and $\Theta$, respectively (see Figure \ref{fig:4.1}). We also report the overall performance based on the whole graph. All simulation results are based on $100$ replications. \\
We fix $p=25$, $d=3$, and consider two sample sizes $n=200$ and $n=300$.
We first generate both $\bm{X}$ and $\bm{Y}$ from multivariate normals.  Specifically, we first generate a $(d+p)\times(d+p)$ sparse precision matrix $\Omega$, in which the probability of the off-diagonal elements being nonzero equals $0.2$. Then $n$ i.i.d. samples $\bm{Z}_1, \dots, \bm{Z}_n$ are generated from $\mathcal{N}(\bm{0},\Omega\inv)$. The decomposition $\bm{Z}_i^T = (\bm{X}_i^T, \bm{Y}_i^T)$ leads to i.i.d. samples of $\bm{X}$ and $\bm{Y}$, and the decomposition $\Omega=\begin{bmatrix}
\Omega_{xx} & \Omega_{xy} \\
\Omega_{yx} & \Omega_{yy}
\end{bmatrix}$ leads to $\Theta = \Omega_{xy}$ and $\Lambda = \Omega_{yy}$. The results are presented in Table \ref{tb:6.3.1}.\\
Overall, the cSScGG and QUIC methods perform better than the NPN. This is expected as the true distribution is Gaussian and the ECDF transformation leads to efficiency loss. 
Surprisingly, the cSScGG outperforms the QUIC in detecting edges within $\bm{X}$ variables even when the normality assumption holds for the QUIC method. It suggests that the proposed projection ratio method learns the conditional independence within  $\bm{X}$ better than the parametric QUIC method with BIC. Furthermore, the cSScGG outperforms the QUIC in identifying edges among $\bm{Y}$ as well as edges between $\bm{X}$ and $\bm{Y}$,  due to the fact that there are two penalty parameters in cSScGG, as opposed to one in QUIC. To conclude, the cSScGG method is more efficient even when the joint normality assumption holds.

\begin{table}[H]
	\centering
	\scalebox{.9}{
		\setlength{\tabcolsep}{0.1em} % for the horizontal padding
		{\renewcommand{\arraystretch}{0.8}% for the vertical padding
			\begin{tabular}{|c|ccc|ccc|ccc|} \hline &  \multicolumn{3}{c|}{cSScGG} & \multicolumn{3}{c|}{QUIC} & \multicolumn{3}{c|}{NPN} \\ 
				& SPE & SEN & F$_1$ & SPE & SEN & F$_1$ & SPE & SEN & F$_1$ \\ \hline 
				\multicolumn{10}{|l|}{\ Among $\bm{X}$} \\ \hline 
				n=200 & 
				\begin{tabular}[c]{@{}c@{}}0.881  \\ (0.221) \end{tabular} & \begin{tabular}[c]{@{}c@{}}0.931  \\ (0.24) \end{tabular} & \begin{tabular}[c]{@{}c@{}}0.732  \\ (0.429) \end{tabular} & \begin{tabular}[c]{@{}c@{}}0.775  \\ (0.245) \end{tabular} & \begin{tabular}[c]{@{}c@{}}0.97  \\ (0.171) \end{tabular} & \begin{tabular}[c]{@{}c@{}}0.55 \\ (0.471) \end{tabular} & 
				\begin{tabular}[c]{@{}c@{}}0.839  \\ (0.259) \end{tabular} & \begin{tabular}[c]{@{}c@{}}0.914  \\ (0.27) \end{tabular} & \begin{tabular}[c]{@{}c@{}}0.699  \\ (0.412) \end{tabular} \\ 
				n=300 & 
				\begin{tabular}[c]{@{}c@{}}0.932  \\ (0.203) \end{tabular} & \begin{tabular}[c]{@{}c@{}}0.895  \\ (0.278) \end{tabular} & \begin{tabular}[c]{@{}c@{}}0.827  \\ (0.352) \end{tabular} & \begin{tabular}[c]{@{}c@{}}0.812  \\ (0.293) \end{tabular} & \begin{tabular}[c]{@{}c@{}}0.989  \\ (0.102) \end{tabular} & \begin{tabular}[c]{@{}c@{}}0.713  \\ (0.428) \end{tabular} &
				\begin{tabular}[c]{@{}c@{}}0.803  \\ (0.304) \end{tabular} & \begin{tabular}[c]{@{}c@{}}0.968  \\ (0.17) \end{tabular} & \begin{tabular}[c]{@{}c@{}}0.765  \\ (0.372) \end{tabular} \\  \hline
				\multicolumn{10}{|l|}{\ Among $\bm{Y}$} \\ \hline  
				n=200 & 
				\begin{tabular}[c]{@{}c@{}}0.821 \\ (0.029)\end{tabular} & \begin{tabular}[c]{@{}c@{}}0.939 \\ (0.039)\end{tabular} & \begin{tabular}[c]{@{}c@{}}0.707 \\ (0.035)\end{tabular}  & \begin{tabular}[c]{@{}c@{}}0.794 \\ (0.03)\end{tabular} & \begin{tabular}[c]{@{}c@{}}0.946 \\ (0.037)\end{tabular} & \begin{tabular}[c]{@{}c@{}}0.682 \\ (0.034)\end{tabular} &
				\begin{tabular}[c]{@{}c@{}}0.819 \\ (0.095)\end{tabular} & \begin{tabular}[c]{@{}c@{}}0.774 \\ (0.333)\end{tabular} & \begin{tabular}[c]{@{}c@{}}0.564 \\ (0.197)\end{tabular} \\ 
				n=300 & 
				\begin{tabular}[c]{@{}c@{}}0.858 \\ (0.028)\end{tabular} & \begin{tabular}[c]{@{}c@{}}0.96 \\ (0.028)\end{tabular} & \begin{tabular}[c]{@{}c@{}}0.761  \\ (0.028) \end{tabular} & \begin{tabular}[c]{@{}c@{}}0.829  \\ (0.027) \end{tabular} & \begin{tabular}[c]{@{}c@{}}0.963  \\ (0.027) \end{tabular} & \begin{tabular}[c]{@{}c@{}}0.728  \\ (0.031) \end{tabular} &
				\begin{tabular}[c]{@{}c@{}}0.79  \\ (0.029)\end{tabular} & \begin{tabular}[c]{@{}c@{}}0.965 \\ (0.024)\end{tabular} & \begin{tabular}[c]{@{}c@{}}0.689 \\ (0.033) \end{tabular}  \\   \hline
				\multicolumn{10}{|l|}{\ Between $\bm{X}$ and $\bm{Y}$} \\ \hline  
				n=200 & 
				\begin{tabular}[c]{@{}c@{}}0.828  \\ (0.117) \end{tabular} & \begin{tabular}[c]{@{}c@{}}0.865  \\ (0.163) \end{tabular} & \begin{tabular}[c]{@{}c@{}}0.687  \\ (0.096) \end{tabular} & \begin{tabular}[c]{@{}c@{}}0.776  \\ (0.06) \end{tabular} & \begin{tabular}[c]{@{}c@{}}0.942  \\ (0.071) \end{tabular} & \begin{tabular}[c]{@{}c@{}}0.656  \\ (0.077) \end{tabular} &
				\begin{tabular}[c]{@{}c@{}}0.821  \\ (0.107) \end{tabular} & \begin{tabular}[c]{@{}c@{}}0.78   \\ (0.337) \end{tabular} & \begin{tabular}[c]{@{}c@{}}0.574  \\ (0.221) \end{tabular} \\ 
				n=300 & 
				\begin{tabular}[c]{@{}c@{}}0.799  \\ (0.125) \end{tabular} & \begin{tabular}[c]{@{}c@{}}0.966  \\ (0.063) \end{tabular} & \begin{tabular}[c]{@{}c@{}}0.707  \\ (0.084) \end{tabular} & \begin{tabular}[c]{@{}c@{}}0.836  \\ (0.051) \end{tabular} & \begin{tabular}[c]{@{}c@{}}0.969  \\ (0.045) \end{tabular} & \begin{tabular}[c]{@{}c@{}}0.728  \\ (0.062) \end{tabular} &
				\begin{tabular}[c]{@{}c@{}}0.786  \\ (0.056) \end{tabular} & \begin{tabular}[c]{@{}c@{}}0.955  \\ (0.059) \end{tabular} & \begin{tabular}[c]{@{}c@{}}0.678  \\ (0.072) \end{tabular} \\   \hline
				\multicolumn{10}{|l|}{\ Overall} \\ \hline  
				n=200 & 
				\begin{tabular}[c]{@{}c@{}}0.823  \\ (0.029) \end{tabular} & \begin{tabular}[c]{@{}c@{}}0.926  \\ (0.044) \end{tabular} & \begin{tabular}[c]{@{}c@{}}0.702  \\ (0.033) \end{tabular} & \begin{tabular}[c]{@{}c@{}}0.79  \\ (0.026) \end{tabular} & \begin{tabular}[c]{@{}c@{}}0.946  \\ (0.036) \end{tabular} & \begin{tabular}[c]{@{}c@{}}0.677  \\ (0.03) \end{tabular} &
				\begin{tabular}[c]{@{}c@{}}0.82   \\ (0.094) \end{tabular} & \begin{tabular}[c]{@{}c@{}}0.776  \\ (0.331) \end{tabular} & \begin{tabular}[c]{@{}c@{}}0.568  \\ (0.194) \end{tabular} \\ 
				n=300 & 
				\begin{tabular}[c]{@{}c@{}}0.848  \\ (0.026) \end{tabular} & \begin{tabular}[c]{@{}c@{}}0.96  \\ (0.028) \end{tabular} & \begin{tabular}[c]{@{}c@{}}0.746  \\ (0.029) \end{tabular} & \begin{tabular}[c]{@{}c@{}}0.831  \\ (0.023) \end{tabular} & \begin{tabular}[c]{@{}c@{}}0.964  \\ (0.026) \end{tabular} & \begin{tabular}[c]{@{}c@{}}0.729  \\ (0.025) \end{tabular} &
				\begin{tabular}[c]{@{}c@{}}0.79  \\ (0.025) \end{tabular} & \begin{tabular}[c]{@{}c@{}}0.963  \\ (0.025) \end{tabular} & \begin{tabular}[c]{@{}c@{}}0.689  \\ (0.027) \end{tabular} \\  \hline
			\end{tabular}
	}}
	\caption{Averages and standard deviations (in parentheses) of specificity(SPE), sensitivity(SEN),  and F$_1$ score when $p = 25$ and $d=3$. $\bm{X}$ follows the multivariate Normal distribution. }
	\label{tb:6.3.1}
\end{table}

%{\color{blue}one possible reason why s.d. among x is so large: sensitivity is TP/P. When the dimension of x is small (like d=3) and graph is sparse, it's likely that TP=0.}

\section{Applications}
\label{sec:application}
\subsection{Isoprenoid Gene Network in Arabidopsis Thaliana}
\label{sec:7.1}
We consider the gene expression data for \textit{Arabidopsis thaliana} introduced by \citeasnoun{wille2004sparse}. 
\textit{Arabidopsis thaliana} is the first plant to have its genome sequenced, and is a popular model in the study of molecular biology and genetics. The dataset contains $n=118$ observations of Affymetrix GeneChip microarrays, in which the expression levels of $795$ genes are recorded. All values are preprocessed by log-transformation and standardization. This data has been analyzed by \citeasnoun{lafferty2012sparse} to explore the structure using the nonparanormal model.  As in \citeasnoun{lafferty2012sparse}, we consider a subset of genes from the isoprenoid pathway \footnote{The dataset was downloaded from \url{https://www.ncbi.nlm.nih.gov/pmc/articles/PMC545783/}. We note that while there were $40$ genes in \citeasnoun{wille2004sparse} and \citeasnoun{lafferty2012sparse}, this dataset contains $39$ only.}.\\
Our goal is to construct a graph using the proposed cSScGG procedure and compare its structure with those from Glasso \cite{friedman2008sparse} and nonparanormal (NPN). Let $\bm{Z}$ be the expression levels of $39$ genes. To apply the cSScGG procedure we first need to identify variables $\bm{X}$ of which the density function may be non-Gaussian. A simple approach is to select elements in $\bm{Z}$ whose marginal distributions are non-Gaussian. We looked at histograms of all $39$ gene expression levels and found $3$ genes (\texttt{MCT}, \texttt{GGPPS6} and \texttt{GGPPS1mt}) with marginal distribution far from Gaussian, as shown in Figure \ref{fig:7.1.1}. Therefore, we set $\bm{X}$ as gene expression levels of \texttt{MCT}, \texttt{GGPPS6}, and \texttt{GGPPS1mt}. We note that marginal distributions of these three genes have bi-/multiple modes, and monotone transformations cannot transfer them into Gaussian random variables. Therefore, the GGM and nonparanormal model may be inappropriate for this data.     \\
\begin{figure}[h!]
	\centering
	\includegraphics[width=0.8\textwidth]{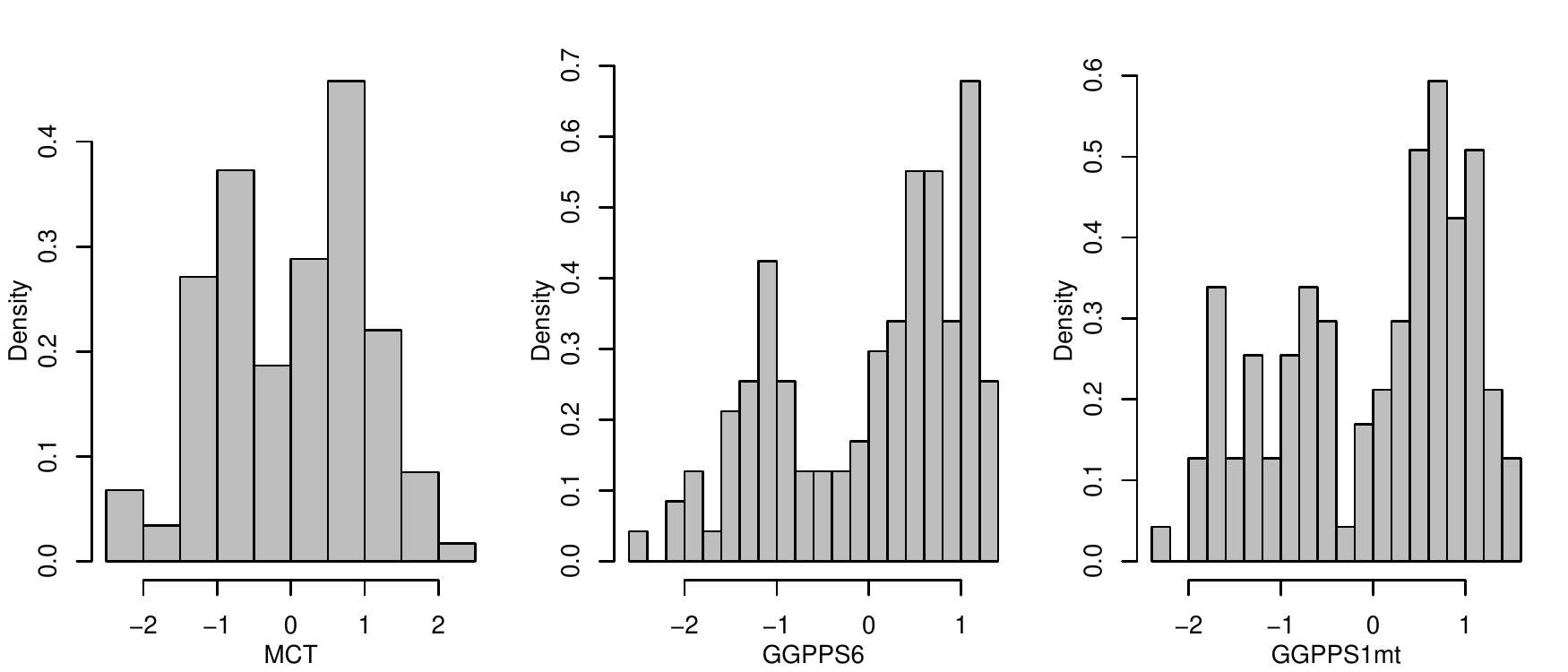}	
	\caption{Histogram of three genes in the gene expression data.}\label{fig:7.1.1}
\end{figure}
As indicated by \citeasnoun{wille2004sparse}, the GGM chosen by the BIC generally leads to a graph that is too dense for biologically relevant researches. Therefore in this study, we construct the graph by limiting the number of edges. Particularly, we tune the regularization parameters in the cSScGG method to fix $|E|=18$.
Results with $|E|=25$ can be found in \citeasnoun{Luothesis}.
Once the cSScGG fit is obtained, we scan the full regularization path of the Glasso estimates, compare the symmetric difference with the cSScGG estimate, and select the graph with smallest symmetric difference value as the Glasso graph. Specifically, the symmetric difference between two graphs is the set of edges which are in either of the graphs but not in their intersection. The same procedure is done for the NPN estimates. We implemented Glasso and NPN with R-packages \texttt{glasso} \cite{friedman2014glasso} and \texttt{huge} \cite{zhao2012huge} respectively.\\
\begin{figure}[h!]
	\centering
	\includegraphics[width=\textwidth]{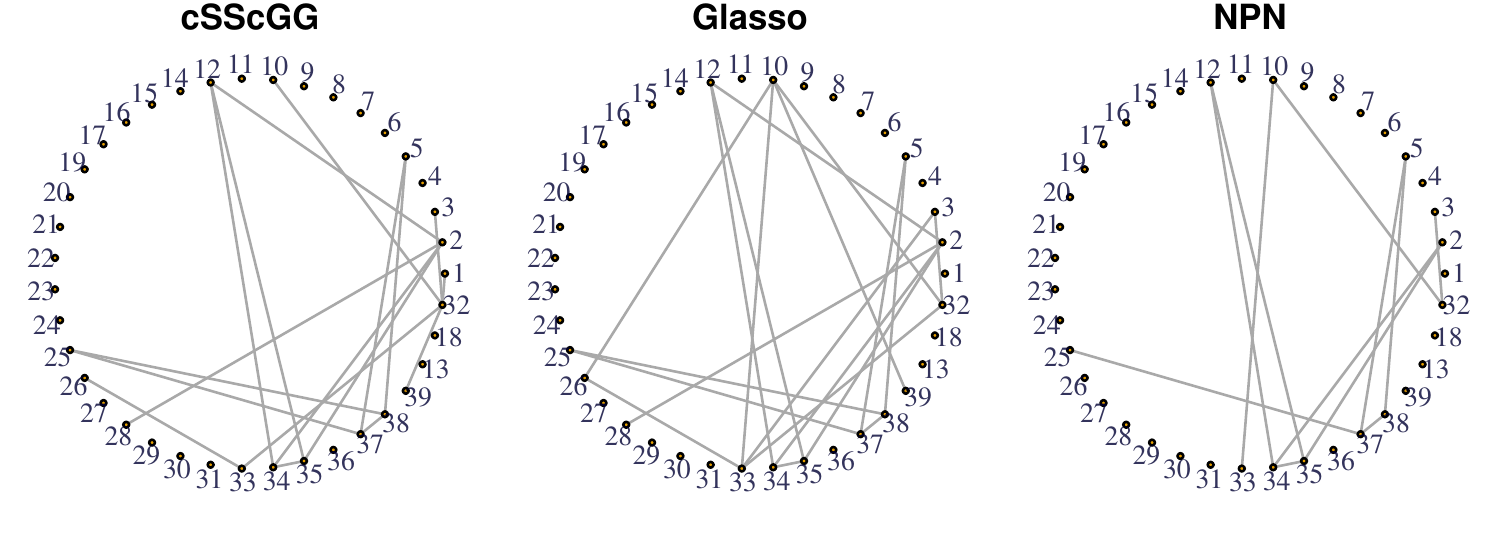}
	\includegraphics[width=0.8\textwidth]{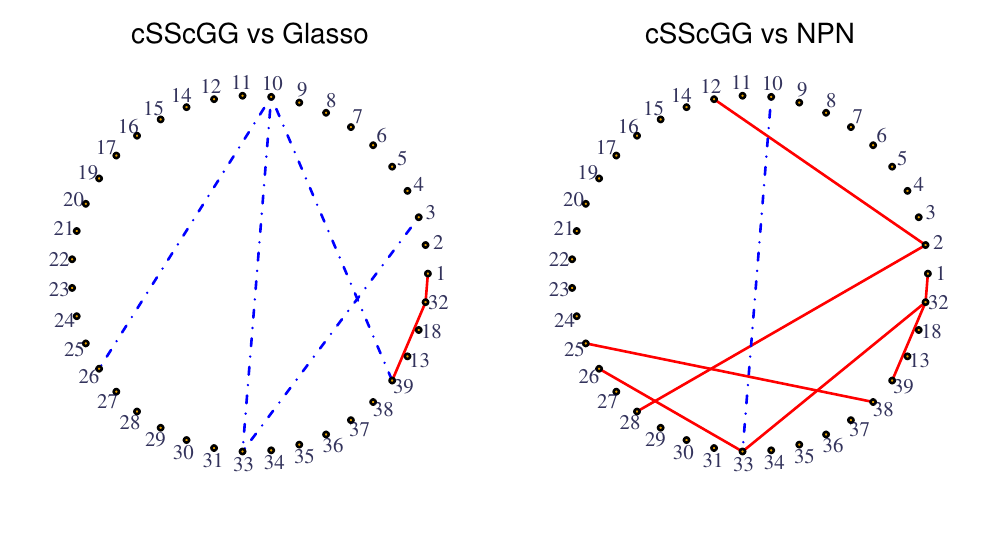}
	\caption{The estimated graph with $18$ edges from the cSScGG (\textit{top left}), the closest Glasso (\textit{top middle}), the closest NPN (\textit{top right}), the symmetric difference between cSScGG and Glasso (\textit{bottom left}), and the symmetric difference between cSScGG and NPN (\textit{bottom right}). Red edges \protect\redline \ in the bottom represent those selected by the cSScGG but not by the Glasso/NPN, blue edges \protect\blueline represent those selected by the Glasso/NPN but not by the cSScGG. Genes \texttt{GGPPS1mt}, \texttt{GGPPS6}, \texttt{MCT} correspond to nodes with numbers $13$, $18$, $32$, respectively.}
	\label{fig:7.1.3}
\end{figure}
Figures \ref{fig:7.1.3} presents graph topologies achieved from each method, along with the corresponding symmetric difference. We refer the symmetric difference between cSScGG and Glasso to as \textit{cSScGG vs Glasso}, and the symmetric difference between cSScGG and NPN to as \textit{cSScGG vs NPN}. Nodes with numbers $13$, $18$, and $32$ correspond to the $3$ non-Gaussian genes \texttt{GGPPS1mt}, \texttt{GGPPS6}, and \texttt{MCT}, respectively.
Although the overall structures of different methods look similar, there are some interesting differences.\\ 
We focus on the
two symmetric difference plots in Figure \ref{fig:7.1.3}. 
Note that red edges are selected by the cSScGG only. Most of these edges are associated with the non-Gaussian nodes, for example, edges \texttt{32-1} and \texttt{32-39}. This indicates that the cSScGG procedure is able to discover new interactions for the non-Gaussian variables. We further look at the red lines that appear only in one of the two symmetric difference plots. It is interesting to see that they all come from the \textit{cSScGG vs NPN} plot, indicating that cSScGG is able to detect some edges selected by Glasso which are not selected by NPN. This is not surprising since the cSScGG method assumes a conditional Gaussian distribution for the parametric component. Finally, we note that as a trade-off for the newly identified interactions, there exists edges that are selected by both Glasso and NPN, but not by cSScGG. For instance, edge \texttt{10-33} with blue dashed line in Figure \ref{fig:7.1.3}.\\
To summarize, in terms of the overall graph structure, the cSScGG procedure is capable of capturing a majority of edges that are detected by the Glasso method. By modeling the distributions of some genes that clearly violate the Gaussian assumption, the proposed method is capable of detecting interactions that are not selected by other methods. These interactions may provide potential research areas for biological study.

\subsection{Conditional Relationship Between Clinical, Laboratory and Dialysis Variables from Hemodialysis Patients}
\label{sec:7.2}
We apply the cSScGG procedure to study the conditional relationships between some clinical, laboratory and dialysis variables collected from hemodialysis patients. All patients who underwent dialysis treatments at the Fresenius Medical Care - North America during 2010-2014 are considered. We include patients who stayed at the same facility throughout the treatments. To avoid large fluctuation in the first year on dialysis, we use the average measurements in the second year on dialysis from patients who survived longer than two years. For homogeneity, we include white, non-diabetic and non-Hispanic patients. After removing missing values, we have $n=2959$ observations (patients) on the following $27$ variables in $3$ categories: \\
\textbf{Clinical variables:} \texttt{age} (years), \texttt{height} (cm), \texttt{weight} (kg), \texttt{bmi} (body mass index, kg/m$^2$), \texttt{sbp} (systolic blood pressure, mmHg), \texttt{dbp} (diastolic blood pressure, mmHg), \texttt{temp} (temperature, Celsius);\\
\textbf{Laboratory variables:} \texttt{albumin} (g/dL), \texttt{ferritin} (ng/mL), \texttt{hgb} (hemoglobin, g/dL), \texttt{lymphocytes} ($\%$), \texttt{neutrophils} ($\%$), \texttt{nlr} (neutrophils to lymphocytes ratio, unitless), \texttt{sna} (serum sodium concentration, mEq/L or mmol/L), \texttt{wbc} (white blood cell, 1000/mc);\\
\textbf{Dialysis variables:} \texttt{qb} (blood flow, mL/min), \texttt{qd} (dialysis flow, mL/min), \texttt{saline} (mL), \texttt{txttime} (treatment time, min), \texttt{olc} (on-line clearance, unitless), \texttt{idwg} (interdialytic weight gain, kg), \texttt{ufv} (ultrafiltration volume, L), \texttt{ufr} (ultrafiltration rate, mL/hr/kg), \texttt{epodose} (erythropoietin dose, unit), \texttt{volume} (L), \texttt{enpcr} (equilibrated normalized protein catabolic rate, g/kg/day), \texttt{ektv} (equilibrated Kt/V, unitless).\\
Note that \texttt{nlr} and \texttt{epodose} have been transformed to make them close to Gaussian. In particular, \texttt{nlr} equals the logarithm of the neutrophils to lymphocytes ratio, and \texttt{epodose} represents the 1/4 power transformation of the actual erythropoietin dose.\\
\begin{figure}[h!]
	\centering
	\includegraphics[width=\textwidth]{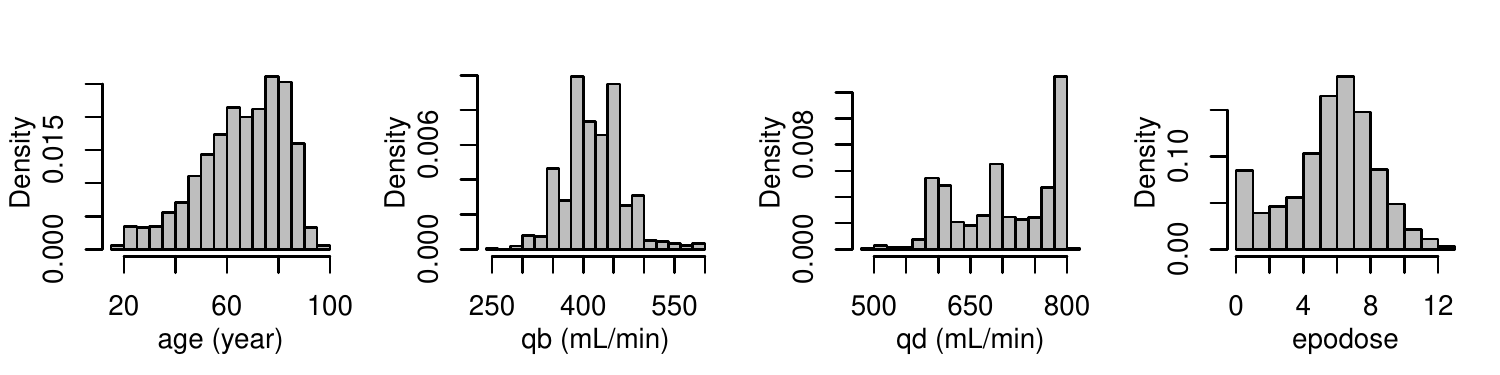}
	\caption{Histograms of \texttt{age}, \texttt{qb}, \texttt{qd} and \texttt{epodose}.}
	\label{fig:7.1}
\end{figure}
%{\color{red} I changed the histogram layout from 2 by 2 to 1 by 4. You can revert to the old layout if space allows.\\ }
The primary objective of this study is to discover the interactions between all these measurements. We first check the marginal distributions of all $27$ variables to investigate possible violation of the Gaussian assumption. We identify $4$ variables, \texttt{age}, \texttt{qb}, \texttt{qd} and \texttt{epodose} as non-Gaussian with very small p-values (less than $2\times 10^{-16}$). Histograms in Figure \ref{fig:7.1} indicate that the distribution of \texttt{age} is skewed, and the distributions of \texttt{qb}, \texttt{qd} and \texttt{epodose} have multiple peaks. Note that despite the $1/4$ power transformation, the distribution of \texttt{epodose} is still far from normal due to the point mass at zero. Consequently, we specify these $4$ variables as $\bm{X}$ to be estimated nonparametrically in the proposed cSScGG procedure.\\
%{\color{red} remove spaceJam result if it is not included in the simulation} {\color{blue} The results shown below are without SpaceJam. Therefore, the unique edges identifed by cSScGG are changed. I changed some comments accordingly.\\}
We compare the cSScGG procedure with Glasso and NPN. For the NPN method, we use the shrunken ECDF to transform the data first, then apply Glasso to the transformed data. For each method, we tune the regularization parameters by BIC. The estimated graph structures are shown in Figure \ref{fig:7.2}.\\
From the visual inspection, there is a large set of edges shared by cSScGG and Glasso, which is due to the fact that cSScGG assumes majority of the variables are conditionally normal. However, the graph of Glasso is much denser. 
To see how cSScGG differs from other two methods, Figure \ref{fig:7.3} shows edges detected by the cSScGG procedure only. It shows that the \texttt{bmi} is a hub node whose connections with other variables such as \texttt{age}, \texttt{dbp}, and \texttt{wbc} are not selected by other methods. Meanwhile, \texttt{qb} has multiple connections with nodes from the other two categories (Clinical and Laboratory).
The value of these extra edges remains to be further explored from a clinical standpoint. 
We do not intend to claim that the graph obtained by the cSScGG procedure is the best as the underlying truth is unknown. Instead, with different model assumptions, the cSScGG procedure can identify potential links for further study.  
%{\color{blue} move enpcr node to right such that it does not overlap with lympho} {\color{red} Sure. I don't have the source code on this laptop, will modify it asap.}

\begin{figure}[h!]
	\centering
	\includegraphics[width=1.05\textwidth]{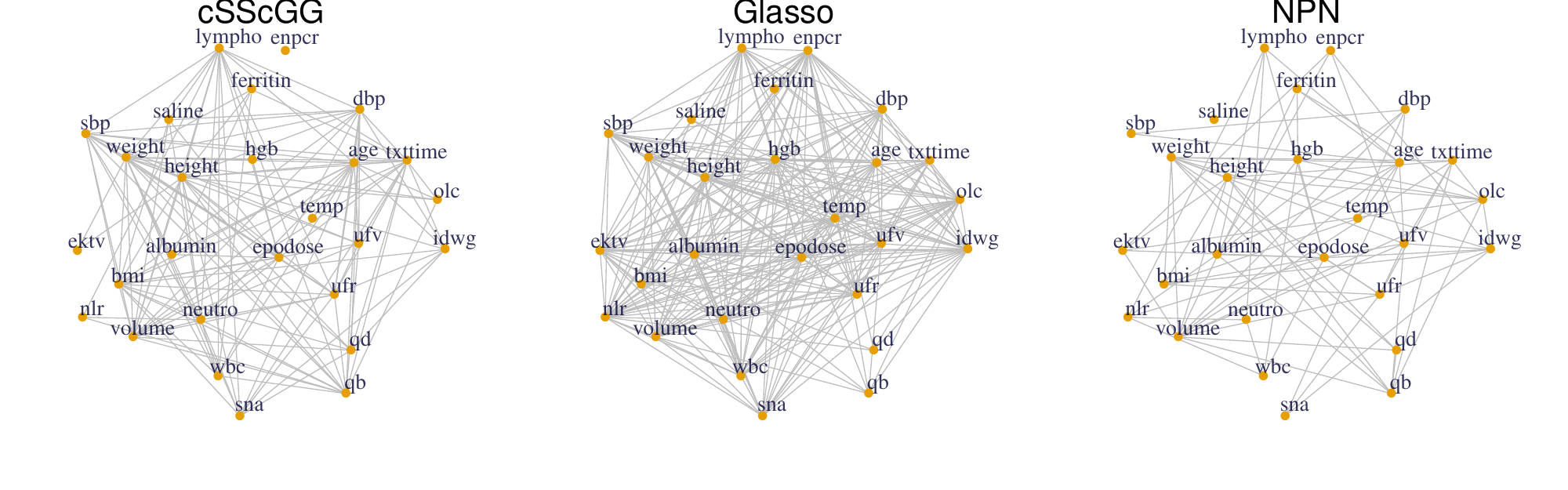}
	\caption{The estimated graphs using cSScGG (\textit{left}), Glasso (\textit{middle}), and NPN (\textit{right}). Tuning parameters are selected by the BIC method. Layout of nodes are fixed across four topologies.}
	\label{fig:7.2}
\end{figure}

\begin{figure}[t]
	\centering
	\includegraphics[width=0.4\textwidth]{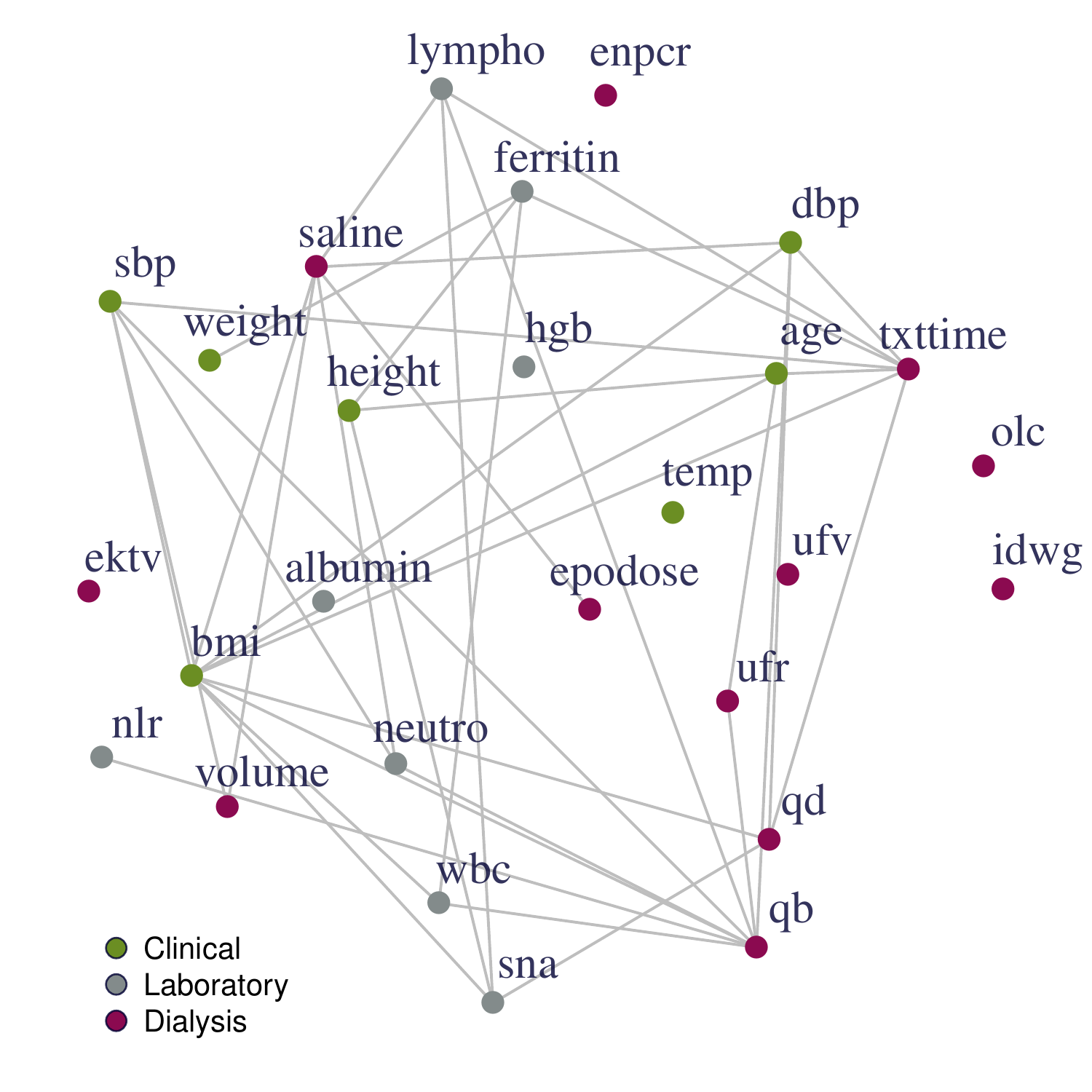}
	\caption{Edges exclusively detected by cSScGG. The measurements are grouped into three categories as described at beginning of Section \ref{sec:7.2}.}
	\label{fig:7.3}
\end{figure}

%\bibliography{reference}
%\bibliographystyle{/home/faculty/yuedong/computer/lib/TeXinputs/dcu}
%\citationstyle{/home/faculty/yuedong/computer/lib/TeXinputs/dcu} 

\newpage
\bibliographystyle{dcu}
\citationstyle{dcu}
\bibliography{reference.bib} 

\newcommand{\Appendix}{\appendix\def\thesection{Appendix~\Alph{section}}
\def\thesubsection{\Alph{section}.\arabic{subsection}}
}

\begin{appendix}
\Appendix    % This makes the section title start with Appendix!

\renewcommand{\theequation}{A.\arabic{equation}}
\renewcommand{\thesubsection}{A.\arabic{subsection}}
\setcounter{equation}{0}

%\clearpage\pagebreak\newpage
%\begin{center}
%S{\scriptsize UPPLEMENTAL MATERIALS: DERIVATION OF POSTERIOR COVARIANCES}
%\end{center}
%{\em The following text is included for refereeing purpose only. It will
%not be part of the paper.}
%\vskip 15pt
%\setcounter{equation}{0}
\section{Derivation of the LOOKL}  \label{appendix:LOOKL}

Our derivation is similar to that in \citeasnoun{lian2011shrinkage} and 
\citeasnoun{vujavcic2015computationally} with adjustments to deal with 
complications brought by the conditional mean $-\Lambda\inv\Theta^T\bx$ 
and two tuning parameters. Recall that a cGGM assumes that 
$\bm{Y}|\bm{X}=\bx \sim \mathcal{N}(-\Lambda\inv\Theta^T\bx,\Lambda\inv)$. 
The log-likelihood based on the $k$-th observation $\bm{X}_k$ and 
$\bm{Y}_k$ is (ignoring constant terms)
\begin{equation}
\tilde{l}_k(\Lambda, \Theta) = \frac{1}{2} \left\{ \log|\Lambda| - 
\tr(\syk\bl + 2 S^T_{xy,k}\bt + \bl\inv\bt^T S^T_{xx,k}\bt) \right\},
\label{eq:3.2.0}
\end{equation} 
where $\syk = \bm{Y}_k^T \bm{Y}_k$, $\sxyk = \bm{X}_k^T \bm{Y}_k$, and 
$\sxk = \bm{X}_k^T \bm{X}_k$ are the empirical variance/covariance matrices. 
Note that $\sy = n^{-1}\sum_{k=1}^{n}\syk$, $\sx = n^{-1}\sum_{k=1}^{n}\sxk$, 
and $\sxy = n^{-1}\sum_{k=1}^{n}\sxyk$.

Let $\hlk$ and $\htk$ be the estimates of $\Lambda$ and $\Theta$ based 
on the data excluding the $k$-th observation.
Directly calculating leave-one-out estimate of the KL distance
is computationally costly. We now derive a score based on the fact that 
cross-validating the log-likelihood provides an estimate of the KL 
distance \cite{yanagihara2006bias}. 

Consider the following function of five variables
$f(\sx,\sy,\sxy,\Lambda,\Theta) = 
\log|\Lambda| - \tr(\sy\Lambda + 2 S^T_{xy}\Theta + \Lambda\inv\Theta^T S^T_{xx}\Theta)$.
We have the identity
$\sum_{k=1}^{n} f(S_{xx,k},S_{yy,k},S_{xy,k},\Lambda,\Theta) = n f(\sx,\sy,\sxy,\Lambda,\Theta)$.
Letting $\bm{S} = (\sx, \sy, \sxy)$ and $\bm{S}_k = (S_{xx,k},S_{yy,k},S_{xy,k})$, 
we denote $f(\sx,\sy,\sxy,\Lambda,\Theta)$ 
and \\
$f(\sxk,\syk,\sxyk,\Lambda,\Theta)$ as $f(\bm{S},\Lambda,\Theta)$ and 
$f(\bm{S}_k,\Lambda,\Theta)$ in the rest of the derivation. The leave-one-out
cross validation score \cite{yanagihara2006bias} 
\begin{align}
&~\text{LOOCV} = -\frac{1}{n}\sum_{k=1}^{n} \tilde{l}_k(\hlk, \htk) 
= -\frac{1}{2n}\sum_{k=1}^{n}f(\bm{S}_k,\hlk,\htk) \notag \\
%&= -\frac{1}{2n}\sum_{k=1}^{n}\Big\{f(\bm{S}_k,\hlk,\htk) - f(\bm{S}_k,\hll,\htt) + f(\bm{S}_k,\hll,\htt) \Big\}\notag\\
=&  -\frac{1}{2} f(\bm{S},\hll,\htt)  -\frac{1}{2n}\sum_{k=1}^{n} \{ f(\bm{S}_k,\hlk,\htk) - f(\bm{S}_k,\hll,\htt) \}
\notag\\
\approx& -\frac{1}{n} l_2(\hll, \htt) -
\frac{1}{2n}\sum_{k=1}^{n}\Big\{ \big(\frac{\partial f(\bm{S}_k,\hll,\htt) }
{\partial \Lambda}\big)^T \vecc({\hlk - \hll}) 
+ \big(\frac{\partial f(\bm{S}_k,\hll,\htt) }
{\partial \Theta}\big)^T \vecc({\htk - \htt} )\Big\},
\label{eq:3.2.2}
\end{align}
where $\partial f(\bm{S}_{k},\hat{\Lambda},\hat{\Theta}) /\partial \Lambda = \partial f(\bm{S}_{k},\hat{\Lambda},\hat{\Theta}) /\partial \vecc(\Lambda)$ 
and 
$\partial f(\bm{S}_{k},\hat{\Lambda},\hat{\Theta}) /\partial \Theta = \partial f(\bm{S}_{k},\hat{\Lambda},\hat{\Theta}) / \partial \vecc(\Theta)$ 
are $p^2$ and $pd$ dimensional column vectors of partial derivatives
given by 
\begin{eqnarray}
\bu_k &\triangleq& \frac{\partial f(\bm{S}_k,\hll,\htt) }{\partial \Lambda} 
= \vecc(\Lambda\inv - \syk + \Lambda\inv\hat{\Theta}^T\sxk\htt\Lambda\inv),
\label{eq:3.2.3} \\
\bw_k &\triangleq& \frac{\partial f(\bm{S}_k,\hll,\htt) }{\partial \Theta} 
= \vecc(-2\sxyk - 2\sxk\Theta\hll\inv).
\label{eq:3.2.4}
\end{eqnarray}
Denoting $\bm{S}^{(-k)}$ as the version of $\bm{S}$ without the $k$-th observation, 
the Taylor expansions of the functions 
$\partial f(\bm{S}^{(-k)},  \hat{\Lambda}^{(-k)}, \hat{\Theta}^{(-k)}) / \partial \Lambda$ 
and 
$\partial f(\bm{S}^{(-k)},  \hat{\Lambda}^{(-k)}, \hat{\Theta}^{(-k)}) /\partial \Theta$ 
at the point $(\bm{S},\hll,\htt)$ are 
\begin{equation}
\begin{aligned}
\bm{0}_{p^2} &=  \frac{\partial f(\bm{S}^{(-k)},\hlk,\htk) }{\partial \Lambda}\\
&\approx \frac{\partial f(\bm{S},\hll,\htt) }{\partial \Lambda} +  
\frac{\partial^2 f(\bm{S},\hll,\htt) }{\partial \Lambda^2} \vecc({\hlk - \hll}) +  
\frac{\partial^2 f(\bm{S},\hll,\htt) }{\partial \Lambda\partial\Theta} \vecc({\htk - \htt})\\
&+ \frac{\partial^2 f(\bm{S},\hll,\htt) }{\partial \Lambda\partial\sx} \vecc({\sxkk - \sx}) + \frac{\partial^2 f(\bm{S},\hll,\htt) }{\partial \Lambda\partial\sy} \vecc({\sykk - \sy})\\
& + \frac{\partial^2 f(\bm{S},\hll,\htt) }{\partial \Lambda\partial\sxy} \vecc({\sxykk - \sxy}),
\end{aligned}
\label{eq:3.2.5}
\end{equation}
and
\begin{equation}
\begin{aligned}
\bm{0}_{pd} &=  \frac{\partial f(\bm{S}^{(-k)},\hlk,\htk) }{\partial \Theta} \\
&\approx \frac{\partial f(\bm{S},\hll,\htt) }{\partial \Theta} +  
\frac{\partial^2 f(\bm{S},\hll,\htt) }{\partial \Theta^2} \vecc({\htk - \htt}) +  
\frac{\partial^2 f(\bm{S},\hll,\htt) }{\partial \Theta\partial\Lambda} \vecc({\hlk - \hll})\\
&+ \frac{\partial^2 f(\bm{S},\hll,\htt) }{\partial \Theta\partial\sx} \vecc({\sxkk - \sx}) + 
\frac{\partial^2 f(\bm{S},\hll,\htt) }{\partial \Theta\partial\sxy} \vecc({\sxykk - \sxy})\\
&+ \frac{\partial^2 f(\bm{S},\hll,\htt) }{\partial \Theta\partial\sy} \vecc({\sykk - \sy}),
\label{eq:3.2.6}
\end{aligned}
\end{equation}
where $\partial^2 f(\bm{S},\Lambda,\Theta)/\partial \Lambda^2 = 
(\partial f(\bm{S},\Lambda,\Theta)/\partial \vecc(\Lambda))/\partial\vecc(\Lambda) $ 
is the $p^2\times p^2$ Hessian matrix, 
$\partial f(\bm{S},\hll,\htt) /\partial \Lambda$ and 
$\partial f(\bm{S},\hll,\htt) /\partial \Theta$ denote partial 
derivative evaluated at $\hat{\Lambda}$ and $\hat{\Theta}$, 
and other second order derivatives are defined similarly. 
Note that $\partial f(\bm{S},\hll,\htt) /\partial \Lambda = \bm{0}$ and 
$\partial f(\bm{S},\hll,\htt) /\partial \Theta = \bm{0}$ because 
$\hat{\Lambda}$ and $\hat{\Theta}$ are the maximum likelihood estimators, 
$\partial^2 f(\bm{S},\hll,\htt)/\partial \Lambda\partial\sxy=\bm{0}$ 
because (\ref{eq:3.2.3}) is free of $\sxy$, and 
$\partial^2 f(\bm{S},\hll,\htt) /\partial \Theta\partial\sy=\bm{0}$ because 
(\ref{eq:3.2.4}) is free of $\sy$.
Let
$A = \partial^2 f(\bm{S},\hll,\htt)/\partial \Theta^2 = -2 \hll\inv\otimes\sx$, 
$B = \partial^2 f(\bm{S},\hll,\htt)/\partial\Theta\partial\Lambda 
= 2\hll\inv\otimes\sx\htt\hll\inv$, 
$C = \partial^2 f(\bm{S},\hll,\htt)/\partial \Lambda^2 
= -\hll\inv\otimes(\hll\inv + 2\hll\inv\hat{\Theta}^T\sx\htt\hll\inv)$, 
$D = \partial^2 f(\bm{S},\hll,\htt) /\partial\Theta\partial\sx = 
-2 \hll\inv\hat{\Theta}^T\otimes I_{d\times d}$, and
$E = \partial^2 f(\bm{S},\hll,\htt)/\partial\Lambda\partial\sx = 
\hll\inv\hat{\Theta}^T\otimes\hll\inv\hat{\Theta}^T$.
Solving (\ref{eq:3.2.5}) and (\ref{eq:3.2.6}) and plugging solutions
into (\ref{eq:3.2.2}), we have 
\begin{eqnarray}
&&\text{LOOKL}(\lambda_2, \lambda_3) \nonumber \\
&=& -\frac{1}{n}l_2(\hll, \htt) + \frac{1}{2n}
\sum_{k=1}^{n}\Big\{ \bu_k^T (-C + B^T A\inv B)\inv 
\big [ (-E + B^T A\inv D)  \bv_{xx,k}
+ 2B^T A\inv \bv_{xy,k} - \bv_{yy,k} \big ] \nonumber \\ 
&& + \ \bw_k^T (-A+BC\inv B^T)\inv 
\big [ (D-BC\inv E) \bv_{xx,k} + BC\inv \bv_{yy,k}
-2 \bv_{xy,k} \big ] \Big\}.
\label{eq:3.2.10}
\end{eqnarray}

For the Gaussian graphical model with 
$\bm{Y}\sim\mathcal{N}(\bm{0},\Lambda^{-1})$, (\ref{eq:3.2.10}) reduces to
\begin{equation*}
-\frac{1}{n} l_2(\hll) + \frac{1}{2n}\sum_{k=1}^{n}\vecc(\Lambda\inv - \syk)^T(\hll\otimes\hll)\vecc(\sykk - \sy),
\end{equation*}
which is the same as the GACV in \citeasnoun{lian2011shrinkage} and KLCV in
\citeasnoun{vujavcic2015computationally}.

\section{Calculation of the Projection Ratio} \label{appendix:projection}

Letting $\hat{\zeta}(\bx) = \hat{\Delta}(\bx) + \hat{\eta}(\bx)$, 
we construct the ratio 
$\tilde{V}(\hat{\zeta} - \tilde{\zeta})/\tilde{V}(\hat{\zeta} - \eta_u)$
where $\tilde{\zeta}$ denotes the squared error projection of 
$\hat{\zeta}$ in $\mathcal{S}^0$. A small ratio indicates that 
$\mathcal{S}^1$ may be removed. By definition, 
\begin{align}
\tilde{V}(\hat{\zeta} - \eta_u)
&=   \intx (\etah + \hat{\Delta} - \etau)^2\rhox\dx - 
\{ \intx(\etah + \hat{\Delta} - \etau)\rhox\dx \}^2 \notag\\
& \triangleq \tilde{V}(\etah - \etau) + \tilde{V}(\hat{\Delta},\hat{\Delta}) + 
2 \tilde{V} (\hat{\eta}-\eta_u,\hat{\Delta}).
\label{eq:4.1.7}
\end{align}

To obtain $\tilde{V}(\hat{\zeta} - \tilde{\zeta})$, one needs to find
\begin{equation}
\tilde{\zeta} =  \argminA_{\zeta =\eta+\hat{\Delta}, \eta \in \mathcal{S}^0} \Big\{ \intx(\hat{\eta} + \hat{\Delta} - \eta)^2(\bx)\rhox\dx - \{\intx (\hat{\eta} + \hat{\Delta} - \eta)(\bx)\rhox\dx\}^2 \Big\}.
\label{eq:4.1.8}
\end{equation}
Let $\mathcal{S}^0=\mathcal{H}^0\oplus\mathcal{H}^1$, where $\mathcal{H}_0$ is a space spanned by known functions
$\{\varphi_1(\bx), \cdots, \varphi_m(\bx)\}$ and $\mathcal{H}_1$ is the orthogonal reproducing kernel Hilbert space with the reproducing kernel
function $R(\cdot,\cdot)$. Let $\bphi=\big(\varphi_i({\bf{X}}_j)\big)_{i=1,\cdots, m}^{j=1,\cdots,n}$ and 
$\bxi=\big(R({\bf{X}}_i,{\bf{X}}_j)\big)_{i=1,\cdots, n}^{j=1,\cdots,n}$.
Let $\tilde{\zeta} = \bphi\trans\tilde{\bdd} + \bxi\trans \tilde{\bcc}$, 
take derivatives with respect to $\tilde{\bdd}$ and $\tilde{\bcc}$, 
and set them to zero. After rearrangements, we obtain the equation
\begin{equation}
\begin{aligned}
&\begin{bmatrix}
\tilde{V}(\bphi,\bphi) &  \tilde{V}(\bphi,\bxi)\\
\tilde{V}(\bxi,\bphi)  &  \tilde{V}(\bxi,\bxi)
\end{bmatrix}
\begin{bmatrix}
\tilde{\bdd} \\ \tilde{\bcc}
\end{bmatrix}
= 
&\begin{bmatrix}
\tilde{V}(\etah+\hat{\Delta},\bphi)\\
\tilde{V}(\etah+\hat{\Delta},\bxi)
\end{bmatrix},
\end{aligned}
\label{eq:4.1.9}
\end{equation}
where $\tilde{V}(\bm{a},\bm{b})={\{ \tilde{V}(a_i,b_j) \}_{i=1}^I}_{j=1}^J$
for any vectors of functions $\bm{a}=(a_1,\ldots,a_I)^T$ and $\bm{b}=(b_1,\ldots,b_J)^T$. 

The right hand side of (\ref{eq:4.1.9}) contains some extra components 
involving $\hat{\Delta}$. We compute solutions to (\ref{eq:4.1.9}) using the Cholesky 
decomposition implemented in the $\texttt{project()}$ function in the 
R package \texttt{gss} \cite{gu2014smoothing}.
Once $\tilde{\zeta}$ is computed, we have 
\begin{equation}
\tilde{V}(\hat{\zeta} - \tilde{\zeta}) = 
\intx(\hat{\zeta} - \tilde{\zeta})^2(\bx)\rhox\dx - 
\big\{\intx (\hat{\zeta} - \tilde{\zeta})(\bx)\rhox\dx\big \}^2 
= \tilde{V}(\hat{\zeta},\hat{\zeta}) + \tilde{V}(\tilde{\zeta},\tilde{\zeta}) 
- 2\tilde{V}(\tilde{\zeta},\hat{\zeta}).
\end{equation}

\section{Proofs of Theoretical Results}  \label{appendix:proofs}

To prove Theorem \ref{thm:1}, 
we first introduce a sequence of lemmas as in \citeasnoun{wytock2013sparse}.
Note that, different from \citeasnoun{wytock2013sparse}, we allow 
different penalties for $\Lambda$ and $\Theta$. 
Lemma \ref{lemma1} below studies the decay rate of the gradients 
$\nabla_\Theta l_2(\Lambda_0,\Theta_0)$ and 
$\nabla_\Lambda l_2(\Lambda_0,\Theta_0)$ in element-wise infinity operator 
norm as sample size increases.
\begin{lemma}
\label{lemma1}
Suppose that the Assumption 1 holds. Then 
\begin{eqnarray}
\mathbb{P}(\norm{\nabla_\Theta l_2(\Lambda_0,\Theta_0)}_\infty > \vartheta) 
&\leq& 2dp\exp\Big\{ -\frac{n\vartheta^2}{8\ccs\cx^2} \Big\}, 
\label{eq:lemma1.1} \\
\mathbb{P}(\norm{\nabla_\Lambda l_2(\Lambda_0,\Theta_0)}_\infty > \vartheta) 
&\leq& 4p^2\exp\Big\{ -\frac{n\vartheta^2}{3200\ccs} \Big\},
\label{eq:lemma1.2}
\end{eqnarray}
for any $\vartheta\in (0,40\cc)$. 
\end{lemma}
\begin{proof}
\citeasnoun{wytock2013sparse} proved (\ref{eq:lemma1.1}) using the 
Chernoff bound for the Gaussian tail probability. 
\citeasnoun{ravikumar2011high} proved (\ref{eq:lemma1.2}) in their Lemma 1.
\end{proof}

The next lemma extends the primal-dual witness approach proposed in 
\citeasnoun{wainwright2009sharp} to our multi-penalties setting.
Let $\Gamma  = (\Lambda^{T}, \Theta^{T})^T$. With a bit abuse of notation, let $l_2(\Gamma)=l_2(\Lambda,\Theta)$.

\begin{lemma}
	Suppose that the true parameter $\Gamma_0$ has support $S$. We consider two optimization problems:
	\begin{equation}
	\hat{\Gamma} = \argminA_{\Gamma} \Big\{l_2(\Gamma) + \lambda(\norm{\Lambda}_1 + r\norm{  \Theta}_1 )\Big\},
	\label{eq:l1}
	\end{equation}
	\begin{equation}
	\tilde{\Gamma} = \argminA_{\Gamma,\Gamma_{\bar{S}=0}} \Big\{ l_2(\Gamma) + \lambda(\norm{\Lambda}_1 + r\norm{  \Theta}_1 )\Big\}.
	\label{eq:l5}
	\end{equation}
	Let $\Delta = \tilde{\Gamma}-\Gamma_0$ and $R(\Delta) = \nabla_\Gamma^2 l_2(\Gamma_0)\Delta + \nabla_\Gamma l_2(\Gamma_0) - \nabla_\Gamma l_2(\tilde{\Gamma})$. If the following conditions hold,
	\begin{enumerate}
		\item the solution $\tilde{\Gamma}$ is unique;
		\item $\vertiii{\big(\nabla^2_\Gamma l_2(\Gamma_0)\big)_{\bar{S}S} \big(\nabla^2_\Gamma l_2(\Gamma_0)\big)_{SS}^{-1}}_{\infty} < 1-\alpha$ for $0<\alpha<1$;
		\item $\max\{ \norm{\nabla_\Gamma l_2(\Gamma_0)}_\infty, \norm{R(\Delta)}_\infty \} \leq \frac{\alpha  \lambda}{8}$;
		%\item The rate $r > 4(1-\alpha)/(\alpha^2 - 2\alpha + 4)$;
	\end{enumerate}
	then the two $\ell_1$-regularized solutions are identical, $\tilde{\Gamma}=\hat{\Gamma}$.
	\label{lemma2}
\end{lemma}
\begin{proof}
Define $\Delta_\Lambda = \tls-\Lambda_0$, $\Delta_\Theta=\tts-\Theta_0$ 
and $\Delta=(\Delta_\Lambda^T,\Delta_\Theta^T)^T$. Let 
$R(\Delta)=(R_\Lambda^T (\Delta_\Lambda,\Delta_\Theta),	R_\Theta^T (\Delta_\Lambda,\Delta_\Theta))^T$ 
be the residual of second order Taylor expansion of the log-likelihood
where 
\begin{eqnarray*}
R_\Lambda (\Delta_\Lambda,\Delta_\Theta) &=& 
\nabla_\Lambda^2 l_2(\Lambda_0,\Theta_0)\Delta_\Lambda +  
\nabla_\Theta\nabla_\Lambda l_2(\Lambda_0,\Theta_0)\Delta_\Theta
+ \nabla_\Lambda l_2(\Lambda_0,\Theta_0)
- \nabla_\Lambda l_2(\Lambda_0+\Delta_\Lambda, \Theta_0+\Delta_\Theta),\\
R_\Theta (\Delta_\Lambda,\Delta_\Theta) &=& 
\nabla_\Theta^2 l_2(\Lambda_0,\Theta_0)\Delta_\Theta +  
\nabla_\Lambda\nabla_\Theta l_2(\Lambda_0,\Theta_0)\Delta_\Lambda + 
\nabla_\Theta l_2(\Lambda_0,\Theta_0)
 - \nabla_\Theta l_2(\Lambda_0+\Delta_\Lambda, \Theta_0+\Delta_\Theta).\notag
\end{eqnarray*}
Following the same arguments as in Lemma 3 in \citeasnoun{ravikumar2011high}, 
the $\ell_1$ optimization problem (\ref{eq:l1}) satisfies
	\begin{equation}
	\nabla^2_\Gamma l_2(\Gamma_0)\Delta + \nabla_\Gamma l_2(\Gamma_0) - R(\Delta) + \lambda Z = 0,
	\label{eq:l2}
	\end{equation}
	where $Z=(Z_\Lambda^T,Z_\Theta^T)^T$ is the sub-differential of the penalty term evaluated at $\Lambda$ and $\Theta$, and
	\begin{equation}
	Z_{\Lambda,ij} = 
	\begin{cases}
	0 & \text{if}\ i=j\\
	\text{sign}(\Lambda_{ij}) & \text{if}\ i\neq j\ \text{and}\ \Lambda_{ij}\neq 0\\
	\in [-1,1] & \text{if}\ i\neq j \ \text{and}\ \Lambda_{ij} = 0,
	\end{cases}\notag
	\end{equation}
	\begin{equation}
	Z_{\Theta,ij} = 
	\begin{cases}
	r\times\text{sign}(\Theta_{ij}) & \text{if}\  \Theta_{ij}\neq 0\\
	\in [-r,r] & \text{if}\  \Theta_{ij} = 0.
	\end{cases}\notag
	\end{equation}
If we can verify the strict dual feasibility $\norm{Z_{\bar{S}}}_\infty \leq 1$, 
then by Lemma 3 in \citeasnoun{ravikumar2011high}, the restricted solution 
$\tilde{\Gamma}$ is an optimal solution to the original $\ell_1$ problem, 
i.e., $\tilde{\Gamma} = \hat{\Gamma}$.

Denoting $H=\nabla^2_\Gamma l_2(\Gamma_0)$ and $G=\nabla_\Gamma l_2(\Gamma_0)$ 
for simplicity, the optimality condition of (\ref{eq:l2}) in terms of $S$ 
and $\bar{S}$ can be rewritten as 
	\begin{equation}
	\begin{bmatrix}
	H_{SS} & H_{S\bar{S}} \\
	H_{\bar{S}S} & H_{\bar{S}\bar{S}}
	\end{bmatrix}
	\begin{bmatrix}
	\Delta_S \\ 0
	\end{bmatrix}
	+
	\begin{bmatrix}
	G_S \\ G_{\bar{S}}
	\end{bmatrix}
	-
	\begin{bmatrix}
	R(\Delta)_S \\ R(\Delta)_{\bar{S}}
	\end{bmatrix}
	+ \lambda
	\begin{bmatrix}
	Z_S \\ Z_{\bar{S}}
	\end{bmatrix}
	= 0.
	\label{eq:l3}
	\end{equation}
	Since $H_{SS}$ is invertible, we have
	\begin{equation}
	\Delta_S = H_{SS}^{-1} (R(\Delta)_S - G_S - \lambda Z_S).
	\label{eq:l4}
	\end{equation}
	Plugging (\ref{eq:l4}) back into the second equation in (\ref{eq:l3}),  we obtain
	\begin{equation}
	\begin{aligned}
	Z_{\bar{S}} =& -\frac{1}{\lambda} H_{\bar{S}S} \Delta_S + \frac{1}{\lambda }(R(\Delta)_{\bar{S}}-G_{\bar{S}}) \\
	=& -\frac{1}{\lambda}H_{\bar{S}S}H_{SS}^{-1}(R(\Delta)_S - G_S) + H_{\bar{S}S}H_{SS}^{-1} Z_S + \frac{1}{\lambda}(R(\Delta)_{\bar{S}}-G_{\bar{S}}).\notag
	\end{aligned}
	\end{equation}
	Taking the $\ell_\infty$ norm of both sides gives
	\begin{equation}
	\norm{Z_{\bar{S}}}_\infty \leq \frac{2-\alpha}{\lambda} (\norm{G}_\infty + \norm{R(\Delta)}_\infty) + (1-\alpha) \leq \frac{2-\alpha}{\lambda} \frac{\alpha \lambda}{4} + (1-\alpha) < 1.\notag
	\end{equation}
	%	The last inequality holds by the fact that rate $r$ is always equal or less than 1.
\end{proof}
Based on Lemma \ref{lemma2}, the solution $\tilde{\Gamma}$ is constructed as a witness to the original unrestricted solution $\hat{\Gamma}$. Then $\tilde{\Gamma}$ inherits many optimality properties from $\hat{\Gamma}$, in terms of the discrepancy to the true $\Gamma_0$ and the recovery of the signed sparsity pattern.
Our next step is to bound the residual term $\norm{R(\Delta)}_\infty$ in terms of $\norm{\Delta}_\infty$.
\begin{lemma}[Control of remainder] \label{lemma3}
Suppose that 
$\norm{\Delta}_\infty \leq \gamma^{-1} \emph{min} \{ 1/(3C_{\Sigma}), 
C_{\Theta}/2 \}$,
then 
\begin{equation}
\norm{R(\Delta)}_\infty \leq 206C_{\Sigma}^4C_{\Theta}^2C_X^2
\gamma^2\norm{\Delta}_\infty^2 .\notag
\end{equation}
\end{lemma}
\begin{proof}
We describe the proof briefly since it follows the same steps as in
\citeasnoun{wytock2013sparse}. Denote second order Taylor expansion 
of a function in terms of its differentials
	\begin{equation}
	\begin{aligned}
	f(X+\Delta) &\approx  f(X) + \vect(\nabla_X f(X))^T \vect(\Delta) + \frac{1}{2}\vect(\Delta)^T (\nabla_X^2f(X))\vect(\Delta)\\
	&\triangleq f(X) + df(X;\Delta) + \frac{1}{2}d^2f(X;\Delta). \notag
	\end{aligned}
	\end{equation}

By the definition of $R(\Delta)$ and the mean value theorem, there exists 
$t\in (0,1)$ such that
$R_\Lambda(\Delta_\Lambda, \Delta_\Theta) = 
d(\nabla_\Lambda l_2(\Lambda_0 + t\Delta_\Lambda, \Theta_0 + t \Delta_\Theta); \Delta_\Lambda, \Delta_\Theta)$
and similarly for $R_\Theta(\Delta_\Lambda, \Delta_\Theta)$.
As expressions of the above second differentials are tedious, we do not 
include them here. However, we note that each term in 
$R_\Lambda(\Delta_\Lambda, \Delta_\Theta)$ and 
$R_\Theta(\Delta_\Lambda, \Delta_\Theta)$ has a quadratic expression in 
$\Delta_\Lambda$ and $\Delta_\Theta$, with at most four 
$(\Lambda_0 + t\Delta_\Lambda)^{-1}$ terms, two 
$(\Theta_0 + t\Delta_\Theta)$ terms and one $\sx$ term. Using the fact that
	\begin{equation}
	\norm{ABC}_\infty \leq \vertiii{(C^T \otimes A)\vect(B)}_\infty \leq \vertiii{C}_1\vertiii{A}_\infty\norm{B}_\infty\notag
	\end{equation}
	for any matrices $A, B, C$ and $\norm{\sx}_\infty \leq C_X^2$, each term in the second differentials is bounded by 
	\begin{equation}
	\cx \vertiii{(\Lambda_0 + t\Delta_\Lambda)^{-1}}^4_{\infty} \vertiii{\Theta_0 + t \Delta_\Theta}^2_1 \vertiii{\Delta}_1^2.
	\label{eq:5.2.4}
	\end{equation}
	For an invertible $\Lambda_0$, since $0<t<1$, it is easy to verify that
\begin{equation}
(\Lambda_0 + t\Delta_\Lambda)^{-1} = 
(I + t\Delta_\Lambda\Lambda_0^{-1})^{-1} \Lambda_0^{-1} 
=\sum_{i=0}^{\infty} (-1)^i (t\Lambda_0^{ -1}\Delta_\Lambda)^i\Lambda_0^{-1}. \notag
\end{equation}
	Then 
	\begin{equation}
	\vertiii{(\Lambda_0 + t\Delta_\Lambda)^{-1}}_\infty \leq \vertiii{\Lambda_0^{-1}}_\infty \sum_{i=1}^{\infty} \vertiii{\Lambda_0^{-1}}^i_\infty \vertiii{\Delta_\Lambda}^i_\infty \leq \frac{C_{\Sigma}}{1-\gamma C_{\Sigma}\norm{\Delta}_\infty} \leq \frac{3C_{\Sigma}}{2}.\notag
	\end{equation}
	Similarly, since $\vertiii{\Delta_\Theta}\leq \gamma \norm{\Delta}_\infty$, we have
	\begin{equation}
	\vertiii{\Theta_0 + t\Delta_\Theta}_1 \leq \vertiii{\Theta_0}_1 + \vertiii{\Delta_\Theta}_1 \leq C_{\Theta} + \gamma \norm{\Delta}_\infty \leq \frac{3C_{\Theta}}{2}.\notag
	\end{equation}
Combining with (\ref{eq:5.2.4}), we obtain 
	\begin{equation}
	\norm{R(\Delta)}_\infty \leq 206C_{\Sigma_0}^4C_{\Theta}^2C_X^2\gamma^2\norm{\Delta}_\infty^2.\notag
	\end{equation}
\end{proof}

\begin{lemma}[Control of $\Delta$]
Suppose that
$u \triangleq 2\kappa_{H} (\max\{ \norm{\nabla_\Theta 
l_2(\Lambda_0,\Theta_0))}_\infty, \norm{\nabla_\Lambda l_2(\Lambda_0,\Theta_0))}_\infty  \} 
+ \lambda)  
\leq \min \{ 1/(3C_{\Sigma}\gamma), C_{\Theta}/(2\gamma),
1/(412\kappa_{H}C_{\Sigma}^4C_{\Theta}^2C_X^2\gamma^2) \}$. Then 
\begin{equation}
\norm{\Delta}_\infty = \norm{\tilde{\Gamma} - \Gamma_0}_\infty \leq u.
\label{eq:l40}
\end{equation}
\label{lemma4}
\end{lemma}
\begin{proof}
Recall that $\Delta = \tilde{\Gamma}-\Gamma_0$, $\tilde{\Gamma}_{\bar{S}} 
= \Gamma_{0,\bar{S}} = 0$, therefore 
$\norm{\Delta}_\infty = \norm{\Delta_S}_\infty$. 
Our goal is to bound the deviation $\Delta$. By (\ref{eq:l4}), we have 
$\Delta_S = H^{-1}_{SS} (R(\Delta)_S - G_S - \lambda Z_S)$.
In the following, we use Brouwer's fixed point theorem on a compact 
set to construct a ball $\mathbb{B}(u)$ that contains $\Delta$. 
Define the $\ell_\infty$-ball
$\mathbb{B}(u) = \{ \Delta | \norm{\Delta_S}_\infty < u \}$
and a continuous map $\mathcal{F}: \Delta_S \rightarrow F(\Delta_S)$ such that
\begin{equation}
F(\Delta_S) = H^{-1}_{SS} (R(\Delta_S) - G_S - \lambda Z_S).
\label{eq:l41}
\end{equation}
Now it suffices to show $F\big(\mathbb{B}(u)\big)\in \mathbb{B}(u)$, 
as this implies there is a solution to the above equation. 
By uniqueness of the optimal solution, we can thus conclude that 
$\Delta$ belongs in this ball.

Taking infinity norm to (\ref{eq:l41}), we have
\begin{equation}
\norm{F(\Delta_S)}_\infty \leq \norm{H^{-1}_{SS} }_\infty \norm{R(\Delta)}_\infty + \norm{H^{-1}_{SS} }_\infty \norm{G + \lambda Z}_\infty.
\label{eq:l42}
\end{equation}
For any $\Pi \in \mathbb{B}(u)$, by Lemma \ref{lemma3}, the 
first term in (\ref{eq:l42}) is bounded by
\begin{equation*}
\norm{H^{-1}_{SS}}_\infty \norm{R(\Pi)}_\infty \leq \kappa_{H} 206C_{\Sigma}^4C_{\Theta}^2C_X^2\gamma^2\norm{\Pi}_\infty^2 \leq \frac{u}{2}.
\end{equation*}
By the definition of radius $u$, the second term in 
(\ref{eq:l42}) is bounded by
\begin{equation*}
\norm{H^{-1}_{SS} }_\infty \norm{G + \lambda Z}_\infty \leq \kappa_{H} (\norm{G}_\infty + \lambda) \leq \frac{u}{2}.
\end{equation*}
Therefore, we have $\norm{F(\Pi)}_\infty \leq u$.
\end{proof}

\noindent
{\bf Proof of Theorem \ref{thm:1}}  
We first show that $\tilde{\Gamma}$ equals the solution to 
original objective function (\ref{eq:obj}) $\hat{\Gamma}$ with 
high probability. Then we proceed with the proof conditioning on this event.

By Lemma 1, we have the element-wise tail conditions for $\Lambda$:
$\mathbb{P}(\max_{i,j}|\nabla_{\Lambda,ij} l_2(\Lambda_0,\Theta_0)| > \delta) 
\leq 1/f_\Lambda(n,\delta)$, 
where $f_\Lambda(n,\delta) = (1/4)  \exp{\big(n\delta^2/(3200\ccs)\big)}$, and 
$\nabla_{\Lambda,ij} l_2(\Lambda_0,\Theta_0)$ denotes the $(i,j)$-th element in 
$\nabla_{\Lambda} l_2(\Lambda_0,\Theta_0)$. For a fixed $n$, denote
\begin{equation}
\bar{\delta}_{f_\Lambda}(n;\omega) = \argmaxA_{\delta}\{ f_\Lambda(n,\delta) < \omega \}.
\label{eq:thm1}
\end{equation}
Similarly, for each fixed $\delta > 0$, denote
\begin{equation}
\bar{n}_{f_\Lambda}(\delta;\omega)= \argmaxA_{n}\{ f_\Lambda(n,\delta) < \omega\}.
\label{eq:thm2}
\end{equation}
By the monotonicity of the function $f_\Lambda(\delta;n)$, it is easy to see that 
\begin{equation}
n > \bar{n}_{f_\Lambda}(\delta;\omega) ~ \text{for some}~ \delta > 0 \ \ \ \Longrightarrow \ \ \ \bar{\delta}_{f_\Lambda}(n;\omega)\leq \delta.
\label{eq:thm6}
\end{equation}
Appling Corollary 1 and Lemma 8 in \citeasnoun{ravikumar2011high}, 
for any $\tau > 2$, we have the control of sampling noise for $\hat{\Lambda}$
\begin{equation}
\mathbb{P}\Big(\norm{\nabla_\Lambda l_2(\Lambda_0,\Theta_0)}_\infty > \bar{\delta}_{f_\Lambda}(n;p^\tau)\Big) \leq \frac{1}{p^{\tau - 2}} \rightarrow 0
\label{eq:thm3}
\end{equation}
where $\bar{n}_{f_\Lambda}(\delta;p^\tau) = 3200\ccs(\tau\log p + \log 4)/\delta^2$
and
$\bar{\delta}_{f_\Lambda}(n;p^\tau) = \sqrt{3200\ccs}\sqrt{(\tau\log p + \log 4)/n}$.
Now we develop the control of sampling noise for $\Theta$. Again, by Lemma 1 we have the element-wise tail probability for $\hat{\Theta}$:
\begin{equation}
\mathbb{P}\big(\max_{i,j}|\nabla_{\Theta,ij} l_2(\Lambda_0,\Theta_0)| > \delta\big) \leq \frac{1}{f_\Theta(n,\delta)}\notag
\end{equation}
where $f_\Theta(n,\delta) = (1/2) \exp \big(n\delta^2/(8\ccs C_X^2)\big)$.\\
Define $\bar{\delta}_{f_\Theta}(n;\omega)$ and $\bar{n}_{f_\Theta}(\delta;\omega)$ similarly to (\ref{eq:thm1}) and (\ref{eq:thm2}). Applying the union bound over all $pd$ entries of the gradient matrix, we obtain that 
\begin{equation}
\mathbb{P}(\max_{i,j}|\big(\nabla_{\Theta,ij} l_2(\Lambda_0,\Theta_0)| > \delta\big) \leq \frac{pd}{f_\Theta(n,\delta)}.\notag
\end{equation}
Let $\delta = \bar{\delta}_{f_\Lambda}\big(n;(pd)^\tau\big)$, then for any $\tau > 1$,
\begin{equation}
\mathbb{P}\Big(\norm{\nabla_{\Theta} l_2(\Lambda_0,\Theta_0)}_\infty > \bar{\delta}_{f_\Theta}(n;(pd)^\tau)\Big) \leq \frac{pd}{f_\Theta\Big(n;\bar{\delta}_{f_\Theta}\big(n;(pd)^\tau\big)\Big)} = \frac{1}{(pd)^{\tau -1}} \rightarrow 0.
\label{eq:thm4}
\end{equation}
The last equality follows the fact that $f_\Theta\Big(n,\bar{\delta}_{f_\Theta}\big(n;(pd)^\tau\big)\Big) = (pd)^\tau$, based on the definition of $\bar{\delta}_{f_\Theta}$.\\

Straightforward calculation shows that $\bar{n}_{f_\Theta}\big(\delta;(pd)^\tau\big)$ and $\bar{\delta}_{f_\Theta}\big(n;(pd)^\tau\big)$ take the forms
\begin{equation}
\bar{n}_{f_\Theta}\big(\delta;(pd)^\tau\big) = 8\ccs C_X^2 \big(\frac{\tau\log(pd) + \log 2}{\delta^2}\big)\notag
\end{equation}
and
\begin{equation}
\bar{\delta}_{f_\Theta}\big(n;(pd)^\tau\big) = \sqrt{8\ccs C_X^2}\sqrt{\frac{\tau\log (pd) + \log 2}{n}}.
\notag
\end{equation}
Denote $\bar{n}_{f_\Gamma} = \max\{\bar{n}_{f_\Lambda},\bar{n}_{f_\Theta}\}$, $\bar{\delta}_{f_\Gamma} = \max\{\bar{\delta}_{f_\Lambda},\bar{\delta}_{f_\Theta}\}$, by (\ref{eq:thm3}) and (\ref{eq:thm4}) we have
\begin{equation}
\mathbb{P}\big(\max\{\norm{\nabla_{\Lambda} l_2(\Lambda_0,\Theta_0)}_\infty, \norm{\nabla_{\Theta} l_2(\Lambda_0,\Theta_0)}_\infty\} < \bar{\delta}_{f_\Gamma}\big) \geq 1 - \big(p^{-(\tau-2)}+ (pd)^{-(\tau-1)}\big).
\label{eq:thm5}
\end{equation}
Specifically, 
\begin{equation}
\begin{aligned}
\bar{\delta}_{f_\Gamma} &= \max\Big\{ \sqrt{3200\ccs}\sqrt{\frac{\tau\log p + \log 4}{n}},  \sqrt{8\ccs\cx^2}\sqrt{\frac{\tau\log (pd) + \log 2}{n}} \Big\} \\
&\leq C_{\sigma}C_X^\star\sqrt{3200} \sqrt{\frac{\tau\log (pd) + \log 4}{n}}
\end{aligned}
\label{eq:thm7}
\end{equation}
where $C_X^\star = \max\{\cx^2,1\}$.\\
Let $\mathcal{A}$ denote the event that $\max\{\norm{\nabla_{\Lambda} l_2(\Lambda_0,\Theta_0)}_\infty, \norm{\nabla_{\Theta} l_2(\Lambda_0,\Theta_0)}_\infty\} < \bar{\delta}_{f_\Gamma}$, (\ref{eq:thm5}) implies that
$\mathbb{P}(\mathcal{A}) \geq 1 - \big(p^{-(\tau-2)}+ (pd)^{-(\tau-1)}\big)$.
Accordingly, we condition on the event $\mathcal{A}$ in the following analysis.

Next, we verify that the third assumption in Lemma 2 holds. Choose the 
(larger) regularization penalty $\lambda = (8/\alpha)\bar{\delta}_{f_\Gamma}$, 
then the first half 
$\norm{\nabla_\Gamma l_2(\Gamma_0)}_\infty \leq \alpha  \lambda/8$ 
is satisfied. It remains to establish the bound 
$\norm{R(\Delta)}_\infty \leq \alpha  \lambda/8$. 
We do so by using Lemmas \ref{lemma4} and \ref{lemma3} consecutively.
Choose
\begin{equation}
\delta = \frac{1}{2\kappa_{H}} \big(1+\frac{8}{\alpha}\big)^{-2} \min\Big\{ \frac{1}{3C_{\Sigma}\gamma}, \frac{C_{\Theta}}{2\gamma}, \frac{1}{412\kappa_{H}C_{\Sigma}^4C_{\Theta}^2C_X^2\gamma^2} \Big \},
\notag
\end{equation}
by our choice of $\lambda$, the minimum bound on $n$ and the monotonicity property (\ref{eq:thm6}) , we have
\begin{equation}
2\kappa_{H} \big(1 + \frac{8}{\alpha}\big)^2 \bar{\delta}_{f_\Gamma}  \leq \min\Big\{ \frac{1}{3C_{\Sigma}\gamma}, \frac{C_{\Theta}}{2\gamma}, \frac{1}{412\kappa_{H}C_{\Sigma}^4C_{\Theta}^2C_X^2\gamma^2}  \Big\}.
\notag
\end{equation}
Applying Lemma \ref{lemma4}, we conclude that 
\begin{equation}
\norm{\Delta}_\infty \leq 2\kappa_{H} \big(1 + \frac{8}{\alpha}\big) \bar{\delta}_{f_\Gamma} \leq 2\kappa_{H} \big(1 + \frac{8}{\alpha}\big)^2\bar{\delta}_{f_\Gamma} \leq \frac{1}{\gamma}\min\Big\{ \frac{1}{3C_{\Sigma}}, \frac{C_{\Theta}}{2} \Big \}.
\label{eq:thm8}
\end{equation}
Then Lemma \ref{lemma3} gives
\begin{equation}
\begin{aligned}
\norm{R(\Delta)}_\infty 
&\leq 206C_{\Sigma}^4C_{\Theta}^2C_X^2\gamma^2\norm{\Delta}_\infty^2
\leq 824 C_{\Sigma}^4C_{\Theta}^2C_X^2\gamma^2\kappa_{H}^2 (1 + \frac{8}{\alpha})^2 \bar{\delta}_{f_\Gamma}^2\\
& = \Big(824 C_{\Sigma}^4C_{\Theta}^2C_X^2\gamma^2\kappa_{H}^2 (1 + \frac{8}{\alpha})^2 \bar{\delta}_{f_\Gamma}\Big)\frac{\alpha\lambda}{8}
 \leq \frac{\alpha\lambda}{8}
\end{aligned}
\notag
\end{equation}
where the final inequality follows from 
the lower bound on sample size $n$, and the monotonicity property (\ref{eq:thm6}).

To summarize, we have shown that condition 3 in Lemma \ref{lemma2} holds. 
Furthermore, a finite $C_X$ implies condition 1, and condition 2 is assumed 
by the Assumption \ref{asmp:2}. These allow us to conclude that 
$\tilde{\Gamma}=\hat{\Gamma}$. By (\ref{eq:thm7}) and (\ref{eq:thm8}), 
the estimator $\hat{\Gamma}$ satisfies the $\ell_\infty$ bound claimed in 
Theorem \ref{thm:1}(a). Moreover, by the bound (\ref{eq:l40}) and the definition 
of $u$ in Lemma \ref{lemma4}, the estimate $\tilde{\Gamma}_{ij}$ cannot differ 
enough from $\Gamma_{0,ij}$ to change sign when condition (\ref{eq:5.2.3}) 
is satisfied. This proves Theorem \ref{thm:1}(b).

\noindent
{\bf Proof of Corollary \ref{col:1}} 
Let $\psi = 2\kappa_{H} (1 + 8/\alpha) C_{\sigma}C_X^\star\sqrt{3200} 
\sqrt{\big(\tau\log (pd) + \log 4\big)/n}$. 
From Theorem \ref{thm:1}, we have  
$\max\Big\{ \norm{\hat{\Lambda}-\Lambda_0}_\infty, \norm{\hat{\Theta}-\Theta_0}_\infty \Big\} \le \psi$
with probability at least $1-\big(p^{-(\tau-2)}+ (pd)^{-(\tau-1)}\big)$. 
Since $\Lambda_0$ has at most $p+s_\Lambda$ non-zeros including diagonal elements and $\Theta_0$ has at most $s_\Theta$ non-zeros elements, we have
\begin{eqnarray*}
\vertiii{\hat{\Lambda}-\Lambda_0}_F &=& 
\Big( \sum_{i=1}^{p} (\hat{\Lambda}_{ii} - \Lambda_{0,ii})^2 + 
\sum_{(i,j)\in E}(\hat{\Lambda}_{ij} - \Lambda_{0,ij})^2 \Big)^{1/2} 
\leq \psi\sqrt{p+s_\Lambda}, \label{eq:c12} \\
\vertiii{\hat{\Theta}-\Theta_0}_F &=&  
\Big(\sum_{(i,j)\in E}(\hat{\Theta}_{ij} - 
\Theta_{0,ij})^{2}\Big)^{1/2} \leq \psi\sqrt{s_\Theta}. \label{eq:c13}
\end{eqnarray*}
Combining above two inequalities leads to the bound in (\ref{eq:col1}).

\begin{lemma}
\label{lemma5}
Suppose that the Assumption \ref{asmp3} holds, then for positive definite 
matrices $\hat{\Lambda}$ and $\Lambda_0$,
\begin{eqnarray*}
P\left(\lambda_\mathrm{min}(\hat{\Lambda}) \geq 0.5 C_L\right) 
&\geq& P\left(\vertiii{\hat{\Lambda}-\Lambda_0}_F \leq 0.5 C_L \right),\\
P\left(\vertiii{\hat{\Lambda}^{-1}-\Lambda_0^{-1}}_F \leq 
\frac{2\sqrt{p}}{C_L^2} \vertiii{\hat{\Lambda}-\Lambda_0}_F \right) 
&\geq& P\left(\vertiii{\hat{\Lambda}-\Lambda_0}_F \leq 0.5 C_L \right).
\end{eqnarray*}
\label{lemma:Lambdainv}
\end{lemma}
\noindent
{\bf Proof of Lemma \ref{lemma5}}
This proof is similar to that for Lemma A.1 in \citeasnoun{fan2011high}. Under the event $\vertiii{\hat{\Lambda}-\Lambda_0}_F \leq 0.5 C_L$, for any vector $\mathbf{v}\in \mathbb{R}^{p}$ with Euclidean norm $\norm{\mathbf{v}}=1$, we have
\begin{equation}
\mathbf{v}^T \hat{\Lambda}\mathbf{v} = \mathbf{v}^T \Lambda_0\mathbf{v} - \mathbf{v}^T (\Lambda_0 - \hat{\Lambda})\mathbf{v} \geq \lambda_\mathrm{min}(\Lambda_0) - \vertiii{\hat{\Lambda}-\Lambda_0}_F \geq 0.5 C_L.\notag
\end{equation}
The inequality holds by the fact that $\vertiii{A}_2 \leq \vertiii{A}_F$ for any $A$. Therefore, $\lambda_\mathrm{min}(\hat{\Lambda}) \geq 0.5 C_L$.\\
Meanwhile, 
\begin{equation}
\begin{aligned}
\vertiii{\hat{\Lambda}^{-1}-\Lambda_0^{-1}}_F &\leq \sqrt{p}\vertiii{\hat{\Lambda}^{-1}(\Lambda_0-\hat{\Lambda} )\Lambda_0^{-1}}_2 \\
&\leq \sqrt{p}\lambda^{-1}_\mathrm{min}(\hat{\Lambda}) \vertiii{\hat{\Lambda}-\Lambda_0}_2  \lambda^{-1}_\mathrm{min}(\Lambda_0) \\
&\leq  \frac{2\sqrt{p}}{C_L^2} \vertiii{\hat{\Lambda}-\Lambda_0}_F.
\end{aligned}
\notag
\end{equation}
The first inequality holds because of submultiplicativity of the $\ell_2$ norm, and $\vertiii{A}_F \leq \sqrt{p}\vertiii{A}_2$ for any matrix $A$.\\

\noindent
{\bf Proof of Theorem \ref{thm:2}}  
Recall that $\text{SKL} \big( f_0(\by|\bx),\hat{f}(\by|\bx) \big)$ 
has an explicit form:
\begin{eqnarray}
&&\text{SKL} \big( f_0(\by|\bx),\hat{f}(\by|\bx) \big) \notag \\
&=& \frac{1}{2}\int_\mathcal{X} \bm{x}^TU^T\hat{\Lambda}U\bm{x}f_0(\bx)d\bx 
+ \frac{1}{2}\int_\mathcal{X} \bm{x}^TU^T\Lambda_0U\bm{x}\hat{f}(\bx)d\bx 
+ \frac{1}{2}\tr\big( \hat{\Lambda}^{-1}\Lambda_0\big) + \frac{1}{2}\tr\big(
\Lambda_0^{-1}\hat{\Lambda} \big) - p\notag \\
&\triangleq& I_1 + I_2 + I_3 + I_4 - p,
\label{eq:thm2.1.1}
\end{eqnarray}
where $U = \hat{\Lambda}^{-1}\hat{\Theta}^T  - 
\Lambda_0^{-1}\Theta_0^T$.
We now derive the upper bound for each of the four terms in 
(\ref{eq:thm2.1.1}) conditioning on the event 
$\vertiii{\hat{\Lambda}-\Lambda_0}_F \leq 0.5 C_L$ and $\vertiii{\hat{\Lambda}^{-1}-\Lambda_0^{-1}}_F \leq 
(2\sqrt{p}/C_L^2) \vertiii{\hat{\Lambda}-\Lambda_0}_F$.

We first derive an bound for $I_1$ using the fact that 
$I_1 \leq 2^{-1} \int_\mathcal{X} \vertiii{\hat{\Lambda}}_2\norm{U\bm{x}}^2f_0(\bx)d\bx$.
Note that 
\begin{eqnarray}
\vertiii{\hat{\Lambda}}_2 = \vertiii{\hat{\Lambda}-\Lambda_0 + \Lambda_0}_2 
\leq \vertiii{\hat{\Lambda}-\Lambda_0}_2 + \vertiii{\Lambda_0}_2 
\leq \vertiii{\hat{\Lambda}-\Lambda_0}_F + C_U.
\label{eq:thm2.1.2}
\end{eqnarray}
Furthermore, since the Frobenius norm for a vector equals its Euclidean norm, we have
\begin{equation}
\begin{aligned}
\norm{\mathbf{U}\bm{x}} &= \norm{(\hat{\Lambda}^{-1}\hat{\Theta}\trans - 
\Lambda_0^{-1}\Theta_0\trans)\bx}_2 \\
&\leq \vertiii{\hat{\Lambda}^{-1}\hat{\Theta}\trans - 
\Lambda_0^{-1}\Theta_0\trans}_2 \norm{\mathbf{x}}_2 \\
&= \vertiii{(\hat{\Lambda}^{-1}-\Lambda_0^{-1})(\hat{\Theta}\trans - \Theta_0\trans + \Theta_0\trans) + \Lambda_0^{-1}(\hat{\Theta}\trans-\Theta_0\trans)}_2\norm{\mathbf{x}}_2\\
&\leq \Big\{\vertiii{\hat{\Lambda}^{-1}-\Lambda_0^{-1}}_F (\vertiii{\hat{\Theta} - \Theta_0}_F + \vertiii{\Theta_0}_F) + \vertiii{\Lambda_0^{-1}}_2\vertiii{\hat{\Theta} - \Theta_0}_F \Big\} \norm{\mathbf{x}}_2 \\
&\leq \Big\{\vertiii{\hat{\Lambda}^{-1}-\Lambda_0^{-1}}_F (\vertiii{\hat{\Theta} - \Theta_0}_F + \vertiii{\Theta_0}_F) + \frac{1}{C_L}\vertiii{\hat{\Theta} - \Theta_0}_F\Big\}\norm{\mathbf{x}}_2 \\
&\leq \Big\{\frac{2\sqrt{p}}{C_L^2} \vertiii{\hat{\Lambda}-\Lambda_0}_F  (\vertiii{\hat{\Theta} - \Theta_0}_F + \vertiii{\Theta_0}_F) + \frac{1}{C_L}\vertiii{\hat{\Theta} - \Theta_0}_F\Big\}\norm{\mathbf{x}}_2
\end{aligned}
\label{eq:thm2.1.3}
\end{equation}
The last inequality holds from Lemma \ref{lemma:Lambdainv}. Combined, we have the upper bound for $I_1$
\begin{eqnarray*}
\begin{array}{l}
 \frac{1}{2}\int_{\mathcal{X}} \Big\{ \vertiii{\hat{\Lambda}-\Lambda_0}_F + C_U\Big\}\Big\{ \frac{2\sqrt{p}}{C_L^2} \vertiii{\hat{\Lambda}-\Lambda_0}_F  (\vertiii{\hat{\Theta} - \Theta_0}_F + \vertiii{\Theta_0}_F) + \frac{1}{C_L}\vertiii{\hat{\Theta} - \Theta_0}_F \Big \}^2 \norm{\mathbf{x}}_2^2 f_0(\bx)d\bx \\
\leq (4p) G C_m \vertiii{\hat{\Lambda}-\Lambda_0}_F^3 \vertiii{\hat{\Theta} - \Theta_0}_F^2,
\end{array}
\end{eqnarray*}
where $G = \max \{\int_\mathcal{X} \norm{\mathbf{x}}_2^2 f_0(\bx)d\bx, 
\int_\mathcal{X} \norm{\mathbf{x}}_2^2 \hat{f}(\bx)d\bx \}
\max\{C_U,1\}\max\{D_{T}^2,1\}/\min\{C_L^4,1\}$ and $C_m = \max\{ \int_\mathcal{X} \bm{x}^T\bm{x}f_0(\bx)d\bx,  \int_\mathcal{X} \bm{x}^T\bm{x}\hat{f}(\bx)d\bx \}$.
As the only difference between $I_1$ and $I_2$ lies in whether the 
expectation is calculated with respect to the true or estimated density, 
this bound also applies to $I_2$.

For $I_3$, note that 
\begin{equation*}
\tr(\hat{\Lambda}^{-1}\Lambda_0) = \tr\Big((\hat{\Lambda}^{-1} - \Lambda_0^{-1} + \Lambda_0^{-1})\Lambda_0\Big)
\leq \vertiii{\hat{\Lambda}^{-1}-\Lambda_0^{-1}}_F^2\vertiii{\Lambda_0}_F^2 + p 
\leq \frac{4pD_L^2}{C_L^4} \vertiii{\hat{\Lambda}-\Lambda_0}_F^2 + p,
\end{equation*}
where the first inequality uses the fact that $\tr(A^T  B)$ is an 
appropriate inner product for symmetric matrices $A$ and $B$, and by 
the Cauchy-Schwarz inequality, 
$\tr(A^T  B) \leq \tr(A^T  A)\tr(B^T  B) = \vertiii{A}_F^2\vertiii{B}_F^2$;
and the second inequality holds by Lemma \ref{lemma:Lambdainv} with probability 1. 
Then
\begin{equation}
I_3 \leq \frac{2pD_L^2}{C_L^4} \vertiii{\hat{\Lambda}-\Lambda_0}_F^2. \notag
\end{equation}
For $I_4$, following similar arguments as above, 
\begin{equation*}
\tr(\Lambda_0^{-1}\hat{\Lambda}) = 
\tr\Big((\hat{\Lambda}-\Lambda_0+\Lambda_0)\Lambda_0^{-1}\Big)
\leq \vertiii{\hat{\Lambda}-\Lambda_0}_F^2\vertiii{\Lambda_0^{-1}}_F^2 + p
\leq \frac{p}{C_L^2}\vertiii{\hat{\Lambda}-\Lambda_0}_F^2 + p.
\end{equation*}
Then by Corollary \ref{col:1} and Lemma \ref{lemma:Lambdainv}, we have $I_1$ and $I_2$
on the order of
$\mathcal{O}\Big(n^{-5/2}p^{5/2}(\log pd)^{5/2}\Big)$, and $I_3$ and $I_4$ on the order of 
$\mathcal{O}\Big(n^{-1}p^2(\log pd)\Big)$.
This proves the claim.

\noindent
{\bf Proof of Theorem \ref{thm:3}} 
The bound of $D\big(f_0(\bz),\hat{f}(\bz)\big)$ in (\ref{thm:3.2}) comes straightforwardly by combing (\ref{eq:thm2.1}) and (\ref{eq:thm3.2}).
However, as the parametric part (\ref{eq:thm2.1}) is conditioning on the event $\vertiii{\hat{\Lambda}-\Lambda_0}_F \leq 0.5 C_L$, a new lower bound for the sample size $n$ needs to be derived such that this condition is always satisfied.\\
By the RHS of upper bound (\ref{eq:col1}) in Corollary \ref{col:1}, we have 
\begin{equation}
n \geq \frac{1600}{C_L^2} \kappa_{H}^2 C_{\sigma}^2 32C_{X}^2 (p+s_\Lambda) \big(1+\frac{8}{\alpha}\big)^4 \big(\tau\log(pd) + \log 4\big).
\label{thm:3.3}
\end{equation}
Combining (\ref{thm:3.3}) with (\ref{eq:thm1.1}) yields (\ref{thm:3.1}) after some simple algebra.
\end{appendix}

\end{document}